\documentclass[preprint,11pt]{article}
\usepackage{fullpage}
\usepackage{amssymb,amsfonts,amsmath,amsthm,amscd,dsfont,mathrsfs}
\usepackage{graphicx,float,psfrag,epsfig,amssymb}
\usepackage[usenames,dvipsnames,svgnames,table]{xcolor}
\definecolor{darkgreen}{rgb}{0.0,0,0.9}
\usepackage[pagebackref,letterpaper=true,colorlinks=true,pdfpagemode=none,citecolor=OliveGreen,linkcolor=BrickRed,urlcolor=BrickRed,pdfstartview=FitH]{hyperref}
\usepackage{wrapfig}
\usepackage{relsize}
\usepackage{color}
\usepackage{pict2e}
\usepackage[tight]{subfigure}
\usepackage{algorithm}
\usepackage[noend]{algorithmic}
\usepackage{caption}
\usepackage{nameref}
\usepackage{makecell}
\usepackage[font={small}]{caption} 

\DeclareMathAlphabet{\mathpzc}{OT1}{pzc}{m}{it}

\newtheorem{propo}{Proposition}[section]
\newtheorem{lemma}[propo]{Lemma}

\newtheorem{coro}[propo]{Corollary}
\newtheorem{thm}[propo]{Theorem}

\newtheorem{remark}[propo]{Remark}

\def\cA{{\cal A}}

\def\cH{{\cal H}}
\def\cC{{\cal C}}

\def\event{\mathcal{E}}
\def\tbeta{\tilde{\beta}}

\def\reals{{\mathbb R}}

\def\ve{{\varepsilon}}
\def\eps{{\varepsilon}}
\def\prob{{\mathbb P}}
\def\E{{\mathbb E}}

\def\tC{\widetilde{C}}

\def\L0{{L_0}}

\def\de{{\rm d}}
\def\<{\langle}
\def\>{\rangle}

\def\hSigma{\widehat{\Sigma}}

\def\F{{\sf F}}
\def\ind{{\mathbb I}}

\def\F{{\sf F}}
\def\normal{{\sf N}}

\def\LOND{\textsc{Lond\,}}
\def\LORD{\textsc{Lord\,}}

\def\event{\mathcal{E}}

\def\v*{v_0}
\def\T*{T_0}

\def\u*{u_0}
\def\F*{F_0}

\definecolor{olivegreen}{rgb}{0,0.6,0.4}

\def\FWER{{\rm FWER}}
\def\FDR{{\rm FDR}}
\def\mFDR{{\rm mFDR}}
\def\FDP{{\rm FDP}}
\def\cH{{\mathcal{H}}}

\def\cU{{\mathcal{U}}}

\def\BH{{\rm BH}}

\def\bx{\bar{x}}

\newcommand{\ajcomment}[1]{}

\makeatletter
\newcommand{\labitem}[2]{%
\def\@itemlabel{\text{#1}}
\item
\def\@currentlabel{#1}\label{#2}}
\makeatother

\addtocontents{toc}{\protect\setcounter{tocdepth}{2}}

\title{On Online Control of False Discovery Rate}

\author{Adel Javanmard
            \footnote{Department of Electrical Engineering, Stanford University and UC Berkeley, supported by a CSoI fellowship. Email: \url{adelj@stanford.edu} }
             \,and Andrea~Montanari 
            \footnote{Department of Electrical Engineering and Department of Statistics, Stanford University. Email: \url{montanar@stanford.edu}}
            }

\begin{document}

\maketitle

\begin{abstract}
Multiple hypotheses testing is a core problem in statistical inference and arises in almost every scientific field.  
Given a sequence of null hypotheses $\cH(n) = (H_1,\dotsc, H_n)$, 
Benjamini and Hochberg~\cite{benjamini1995controlling} introduced the false discovery rate (\FDR), which is the expected proportion
of false positives among rejected null hypotheses, and
proposed a testing procedure that controls $\FDR$
below a pre-assigned significance level. They also proposed a different criterion, called $\mFDR$, which does not control a property of
 the realized set of tests; rather it controls the ratio of expected number of false discoveries to the expected number of discoveries.
 
 In this paper, we propose two procedures for multiple hypotheses
 testing that we will call $\LOND$ and $\LORD$. 
These procedures control $\FDR$ and $\mFDR$ in an \emph{online manner}.
 Concretely, we consider an ordered --possibly infinite-- sequence of null hypotheses $\cH =
 (H_1,H_2,H_3,\dots )$ where, at  each step $i$, the
 statistician must  decide whether to reject  hypothesis 
$H_i$ having access only to the previous decisions. To the best of our
knowledge, our work is the first that controls $\FDR$ in this
setting. 
This model was introduced by Foster and Stine~\cite{alpha-investing}
whose alpha-investing rule only controls $\mFDR$ in online manner.
 
In order to compare different procedures, we develop lower bounds on
the total discovery rate under the mixture model where each null
hypothesis is truly false with probability $\eps$, for a fixed
arbitrary $\eps$, independently of others. Conditional on the set of
true null hypotheses, $p$-values are independent, and  iid according
to some non-uniform distribution for the non-null hypotheses. 
Under this model, we prove that both $\LOND$ and $\LORD$ have nearly
linear number of discoveries.
We further propose an adjustment to $\LOND$ to address arbitrary
correlation among the $p$-values. 

Finally, we evaluate the performance of our procedures on both
synthetic and real data comparing them with alpha-investing rule, 
Benjamin-Hochberg method and a Bonferroni procedure.          
\end{abstract}

\section{Introduction}
The common practice in claiming a scientific discovery is to support such claim with a $p$-value
as a measure of statistical significance. Hypotheses with $p$-values below a significance level $\alpha$, typically $0.05$, are considered
to be \emph{statistically significant}. While this ritual controls type I errors for single testing problems, in case of testing multiple
hypotheses we need to adjust the significance levels to control other metrics such as family-wise error rate  (FWER) or false discovery rate (FDR). 

As a concrete example, consider a microarray experiment that studies the changes in genetic expression levels of thousands of genes between a set of normal control samples and  a set of prostate cancer patients.  
The goal is to identify a subset of genes that have association with the prostate cancer. 
This can be formulated as a multiple hypotheses testing problem with
many hypotheses-- say $p$-- where a few of them -- say $s$-- are non
null. Indeed, we expect only a small number of genes to be relevant to the cancer. In such setting, an unguarded use of single-inference procedures leads to a large false positive (false discovery) rate. In particular, if we consider a fixed significance level $\alpha$ for all the tests, each of $p-s$ truly null hypotheses can be falsely rejected with probability $\alpha$.  Therefore, we get $\alpha(p-s)$ wrong findings in expectation. This situation becomes more dramatic when more genes are tested over time and $p/s \to \infty$. 

Let us stress two challenges that arise with increasing frequency in
modern data-analysis problems:
\begin{enumerate}
\item[$I.$] \emph{The number of
hypotheses is unknown or potentially infinite.} This is especially the
case when a given line of research engages numerous teams across the
world. Hypotheses are generated over time by different researchers and
tested without central control. If each test is carried out without
taking into account previous discoveries, this will generate a
constant stream of false discoveries. If the underlying number of true facts is bounded, false
discoveries will over-run true ones over time.

 In the prostate cancer example, more genes factors (genetic and
 environmental)  will be tested for having significant association with cancer. 
 If previous discoveries are not taken into account, this
 will generate a constant stream of false associations.
\item[$II.$] \emph{Scientific research is decentralized.} 
An easy solution to the previous problem could be obtained by
coordinating research on a given topic through a central control. 
In our running example, there could be a center that coordinates
research on prostate cancer. After accumulating all raw experimental data on
the issue, and all hypotheses (e.g. all conjectured associations)
the center carries out the data analysis. For instance, it performs a
multiple hypotheses test, controlling FDR. 

Of course, the very decentralized nature of scientific research
prevents such a solution. We instead seek a solution by which
at each step the statistician can decide whether to reject the current
hypothesis on the basis of the current evidence, and minimal
information about previous hypotheses. 
\end{enumerate}
These remarks motivate the following setting, first introduced in \cite{alpha-investing} (a more formal definition
will be provided below).
\begin{quote}
\emph{Hypotheses arrive sequentially in a stream. At each step, the investigator must decide whether to reject the current
null hypothesis without having access to the number of hypotheses (potentially infinite) or the future $p$-values, but solely based on the previous decisions.}
\end{quote}
In order to illustrate this scenario, consider an approach that would 
control FWER, i.e. the probability of rejecting at least one true
null hypothesis. This can be achieved by choosing different
significance levels $\alpha_i$ for tests 
$H_i$, with $\underline{\alpha}=(\alpha_i)_{i\ge 1}$ summable, e.g.,
$\alpha_i = \alpha 2^{-i}$. 
Notice that the researcher only needs to know the number of tests
performed before the current one, in order to implement this scheme.
However, this leads to small statistical power. 
In particular, obtaining a discovery at later steps becomes very unlikely.

Since  Benjamini and Hochberg's seminal
paper~\cite{benjamini1995controlling}, FDR has been widely used in multiple testing
problems and nowadays serves as the acceptable criterion to reduce
risk of spurious discoveries. FDR controls the proportion of false
rejections rather than the probability 
of at least one rejection. This metric is particularly useful when
there is no strong interest in any single hypotheses, but instead we
would like to find a set of potential predictors. 
Benjamini and Hochberg also proposed a sequential testing procedure, referred to as BH hereafter, to control FDR in multiple testing problems assuming
that all the $p$-values are given a priori. 
Let us briefly recall the BH procedure. Given $p$-values $p_1, p_2, \dotsc, p_n$ and a significance level $\alpha$, 
follow the steps below:
\begin{enumerate}
\item Let $p_{(i)}$ be the $i$th $p$-value in the (increasing) sorted order, and define $p_{(0)} = 0$.  Further. let
\begin{align}\label{eq:iBH}
i_{\BH} \equiv \max\Big\{0\le i\le n:\, p_{(i)}\le \alpha i/n \Big\}\,.
\end{align}
\item Reject $H_j$ for every test where $p_j \le p_{i_{(\BH)}}$.
\end{enumerate}
Note that BH requires the knowledge of \emph{``all"} $p$-values to
determine the significance level for testing the hypotheses. Hence, it
does not address the scenario described above.

In this paper, we propose a method for \emph{online} control of false discovery rate. 
Namely, we consider a sequence of hypotheses $H_1, H_2, H_3, \dotsc$
that arrive sequentially 
in a stream, with corresponding p-values $p_1$, $p_2$, $\dots$. 
We aim at developing a testing mechanism that ensures false discovery
rate remains below 
a pre-assigned level $\alpha$. 
A testing procedure provides a sequence of significance levels $\alpha_i$, with decision rule:
\begin{eqnarray}\label{eq:Ti}
T_i = \begin{cases}
1,& \text{if }p_i\le\alpha_i\quad\quad\;\text{(reject $H_i$)}\\
0,& \text{otherwise}\quad \quad \text{(accept $H_i$)}
\end{cases}
\end{eqnarray}
In \emph{online} testing, we require significance levels to be functions of prior outcomes:
\begin{eqnarray}
\alpha_i = f(\{T_1, T_2, \dotsc, T_{i-1}\})\,.
\end{eqnarray} 

One further motivation for online hypothesis testing is that it allows
to exploit domain knowledge in a more flexible manner. 
In standard multiple hypotheses testing, domain expertise can be used
to choose the collection of (null) hypotheses that are likely to be rejected. However, after hypotheses are chosen there is no exploitation of domain knowledge. In contrast, in an online framework, domain knowledge can be used to order the hypotheses to increase statistical power. 
Our proposed methods have the property that if the hypotheses that are more likely to be rejected appear first, or if they arrive in batches, we gain larger statistical power and higher discovery rate.

Foster and Stein~\cite{alpha-investing} proposed the
alpha-investing method 
that controls a modified measure, called mFDR, in  online multiple
hypothesis testing (under some technical assumptions). $\mFDR$ is the ratio of expected number of
false discoveries to the expected 
number of discoveries. Alpha-investing starts with an initial wealth,
at most $\alpha$, of allowable mFDR rate.  The wealth is spent for
testing different hypotheses.  Each time a discovery occurs, the
alpha-investing procedure earns  a contribution toward its wealth to
use for further tests.
 We refer to Section~\ref{sec:comparison-alpha} for a more detailed
 discussion on alpha-investing 
procedure and comparison with our proposals.    

The contributions of this paper are summarized as follows.
\begin{description}
\item[Online control of FDR and mFDR.] We present two algorithms, that
  we call
  $\LOND$ and $\LORD$, to control FDR and mFDR in an online fashion. The $\LOND$ algorithm sets the significance levels $\alpha_i$ based on the number of discoveries made so far, while $\LORD$ sets them according to the time of the most recent discovery. To the best of our knowledge, this is the first work that guarantees online control of false discovery rate (FDR). We further propose an adjustment to $\LOND$ to cope with arbitrary dependency among the $p$-values.
Note that in many scenarios, the investigator chooses the hypotheses based on the previous decisions and hence the test statistics and the resulting $p$-values are dependent.

\item[Computing total discovery rate.] In order to compare different procedures, we develop lower bounds on the total discovery rate of our methods under the mixture model where each null hypotheses is truly false with probability $\eps$, for a fixed arbitrary $\eps$, independently of other hypotheses. Conditional on the set of true null hypotheses, $p$-values
 are independent with uniform distribution for null hypotheses and are
 iid according to some non-uniform distribution for the non-null
 hypotheses. Under this model, we show that both \LOND and \LORD
 achieve a nearly  linear number of discoveries. 
  
 \item[Numerical Validation.] We validate our procedures on synthetic
   and real data in Section~\ref{sec:numerical}, showing that they control FDR
   and mFDR in an online setting. We further compare them with the
   alpha-investing method~\cite{alpha-investing}, and with BH and
   Bonferroni procedures. 
We observe that, our online procedures are nearly as powerful as BH,
with often much smaller FDR.

In addition, we corroborate our results regarding the total discovery rate.
 \end{description}

In the rest of the introduction, we provide definitions of FWER, FDR and mFD, and discuss related work. In Section~\ref{sec:main results}, we present our procedures, \LORD and \LOND, that control FDR and mFDR in an online manner. Section~\ref{sec:comparison-alpha} describes alpha-investing rules proposed by~\cite{alpha-investing} for controlling mFDR and explains the differences between these rules and our procedures.
We discuss in Section~\ref{sec:domain-knowledge} how \LORD and \LOND leverage the domain knowledge to achieve higher statistical power. In Section~\ref{sec:discovery-rate}, we compute discovery rate of $\LORD$ and $\LOND$ algorithms under the mixture model, showing their order optimality.
We evaluate performance of \LORD, \LOND, Bonferroni, BH and an alpha-investing rule on synthetic examples in Section \ref{sec:numerical}.
Proof of main theorems and lemmas are provided in Section~\ref{sec:proofs}, with several technical steps deferred to Appendices. 
\subsection{Different criteria: FWER, FDR and mFDR}
Consider an ordered sequence of null hypotheses $\cH = (H_i)_{i\ge 1}$, where $H_i$ concerns the value of a parameter $\theta_i$. Without loss of generality, assume that $H_i= \{\theta_i = 0\}$. Rejecting null hypothesis $H_i$ means that $\theta_i$ is \emph{discovered} to be significant.
Let $\Theta$ denote the set of possible values for the parameters. We further let $p_i$ be the $p$-value of test $H_i$ whose distribution depends on the value $\theta_i$. Under the null hypothesis $H_i : \,\theta_i =0$, the corresponding $p$-value is uniformly random in $[0,1]$:
$$p_i\sim \cU([0,1])\,.$$
We let $\cH(n) = (H_1,\dotsc, H_n)$ be the collection of the first $n$ hypotheses in the stream.
The statistic $T_i$ is the indicator that a discovery occurs at time $i$ and $D(n)$ denotes the number of discoveries in $\cH(n)$, hence
$$D(n) = \sum_{i=1}^n T_i\,.$$ 
We also let the random variable $F^\theta_i$ be the indicator that a false discovery occurs at time step $i$ and $V^\theta(n)$ be the number of false discoveries in $\cH(n)$, i.e., the number of hypotheses that are incorrectly rejected. Therefore,
$$V^{\theta}(n) = \sum_{i=1}^n F^\theta_i\,.$$
Throughout the paper, superscript $\theta$ is used to distinguish unobservable variables such as $V^\theta(n)$, from statistics such as  
$D(n)$. However, we drop the superscript when it is clear from the context. 

There are various criteria relevant for multiple testing problem.
\begin{itemize}
\item[$\bullet$] \emph{Family-wise error rate (FWER):} The probability of falsely rejecting any of the null hypotheses in $\cH(n)$:
\begin{equation}
\FWER(n) \equiv \sup_{\theta \in \Theta}\, \prob_\theta\Big(V^\theta(n) \ge 1 \Big)\,.
\end{equation}
\item[$\bullet$] \emph{False discovery rate (FDR):} This criterion was introduced by Benjamini and Hochberg~\cite{benjamini1995controlling}, and is the expected proportion of false discoveries among the rejected hypotheses. We first define
\emph{false discovery proportion} (\FDP) as follows. For $n\ge 1$,
\begin{eqnarray*}
\FDP^\theta(n) \equiv \frac{V^{\theta}(n)}{D(n)\vee 1}\,.
\end{eqnarray*}
The false discovery rate is defined as
\begin{eqnarray}
\FDR(n) \equiv \sup_{\theta\in \Theta} \E_\theta\Big(\FDP^{\theta}(n)\Big)\,.
\end{eqnarray}
\item[$\bullet$] \emph{m-False discovery rate (mFDR):} The
ratio of expected number of false rejections to the expected number of rejections:
\begin{equation}
\mFDR_\eta(m) \equiv \sup_{\theta\in \Theta} \, \frac{\E_\theta(V^\theta(n))}{\E_\theta(D(n))+\eta}\,.
\end{equation}
%
\end{itemize}
Note that while $\FDR$ controls a property of the realized set of tests, $\mFDR$ is the ratio of two expectations
over many realizations.
In general the gap between these metrics can be 
significant. (See Figures~\ref{fig:FDR_Worst_dep2} and~\ref{fig:mFDR_Worst_dep2}.)

\subsection{Further related work}
\label{sec:related work}

We list below a few lines of research  that are related to our work.

\bigskip

\noindent\emph{General context.} An increasing effort was devoted to reducing the risk of fallacious research findings.
Some of the prevalent issues such as publication bias, lack of
replicability and multiple comparisons on a dataset were
discussed in Ioannidis's 2005
papers~\cite{ioannidis2005most,ioannidis2005contradicted} 
and in~\cite{prinz2011believe}. 

\bigskip

\noindent\emph{Statistical databases.} Concerned with the above issues and the
importance of data sharing in the 
genetics community, ~\cite{rosset2014novel} proposed an approach to
public database 
management, called Quality Preserving Database (QPD). 
The premise of QPD is to make a shared data resource amenable to
perpetual use for hypothesis 
testing while controlling FWER and maintaining statistical power of
the tests.  In this scheme, for testing 
a new hypothesis, the investigator should pay a price in form of
additional samples that  should be added to the database. The number
of required samples for each test 
depends on the required effect size and the power for the corresponding test. A key feature of QPD is that controlling type I error is performed at the management layer and the investigator is not concerned with $p$-values for the tests. Instead, investigators provide effect size, assumptions on the distribution of the data, and the desired statistical power. A critical limitation of QPD is that all samples, including those currently in the database and those that will be added, are assumed to have the same quality and are coming from a common underlying distribution.  Motivated by similar concerns in practical data analysis, ~\cite{dwork2014preserving}  applies insights from differential privacy to efficiently use samples to answer adaptively chosen estimation queries.   
These papers however do not address the problem of controlling FDR in online multiple testing.

\bigskip

\noindent\emph{Online feature selection.} Building upon
alpha-investing procedures,~\cite{lin2011vif} develops 
VIF, a method for feature selection in large regression problems. VIF is accurate and computationally very efficient; it uses a one-pass search over the pool of features and applies alpha-investing to test each feature for adding to the model. VIF regression avoids overfitting leveraging the property that alpha-investing controls $\mFDR$. Similarly, one can incorporate $\LORD$ and $\LOND$ procedures in VIF regression to perform fast online feature selection and provably avoid overfitting.

\bigskip

\noindent\emph{High-dimensional and sparse regression.} There has been 
significant interest over the last two years in developing hypothesis
testing procedures for high-dimensional regression, especially in
conjunction with sparsity-seeking methods. Procedures for computing
$p$-values of low-dimensional coordinates were developed in 
\cite{zhang2014confidence,van2014asymptotically,javanmard2014confidence,javanmard2013hypothesis,javanmard2013nearly}.
Sequential and selective inference  methods were proposed in 
\cite{lockhart2014significance,fithian2014optimal,taylor2014exact}.
Methods to control FDR were put forward in \cite{barber2014controlling,bogdan2014slope}.

As exemplified by VIF regression, online hypothesis testing methods can be useful
in this context as they allow to select a subset of regressors through
a one-pass procedure. Also they can be used in conjunction with the
methods of \cite{lockhart2014significance}, where a sequence of
hypothesis is generated by including an increasing number of
regressors (e.g. sweeping values of the regularization parameter).

To the best of our knowledge, the only procedure that compares with
the ones we develop is the ForwardStop rule of
\cite{g2013sequential}.
Note, however, that this approach falls short of addressing the issues 
we consider, for several reasons.
$(i)$ It is not online, at least in the form presented in
\cite{g2013sequential} since it reject the first $\hat{k}$ null
hypotheses, where $\hat{k}$ depends on all the $p$-values.
 $(ii)$ It requires knowledge of all
past $p$-values (not only discovery events) to compute the current
score. $(iii)$ Since it is constrained to reject all hypotheses before
$\hat{k}$, and accept
them after, it cannot achieve any discovery rate
increasing with $n$, let alone nearly linear in $n$. 
For instance in the mixture model of Section \ref{sec:discovery-rate},
if the fraction of true non-null is $\eps<\alpha$, then ForwardStop
achieves $O(1)$ discoveries out of $\Theta(n)$ true non-null. In other
words its power is of order $1/n$ in this simple case (no matter what
is the strength of the signal for non-null hypotheses).

\subsection{Notations}\label{sec:notations}
For two functions $f(n)$ and $g(n)$, the notation $f(n) = \Omega(g(n))$ means that $f$ is bounded below by $g$ asymptotically, namely,
there exists positive constant $C$ and $n_0>0$ such that $f(n) \ge Cg(n)$ for $n > n_0$. The notation $f(n) = \Theta(g(n))$ indicates that
$f$ is bounded both above and below by $g$ asymptotically, i.e., for some $C_1, C_2 >0$ and some positive integer $n_0$ we have $C_1g(n) \le f(n) \le C_2 g(n)$ for $n > n_0$. Throughout, $\phi(t) = e^{-t^2/2}/\sqrt{2\pi}$, $\Phi(t) = \int_{-\infty}^t \phi(t) \de t$. We also
use $x\vee y = \max(x,y)$ and $x\wedge y = \min(x,y)$.

\section{Main results}\label{sec:main results}
We present two algorithms for online control of $\FDR$ and $\mFDR$. The first algorithm sets the significance levels
for tests based on total number of discoveries made so far. We name this algorithm \LOND which stands for (significance) Levels based On Number of Discoveries. The second algorithm sets the significance level at each step based on the time the last discovery has occurred. We name this algorithm \LORD which stands for (significance) Levels based On Recent Discovery.

\subsection{{\sc \LOND} algorithm }
We choose any sequence of nonnegative numbers $\underline{\beta} = (\beta_i)_{i=1}^\infty$, such that $\sum_{i=1}^\infty \beta_i = \alpha$.
The values of significance levels $\alpha_i$ are chosen as follows:
\begin{eqnarray}\label{eq:RuleII}
\alpha_i = \beta_i (D(i-1)+1)\,.
\end{eqnarray}
\begin{thm}\label{thm:FDRII}
Suppose that conditional on $D(i-1)$, we have
\begin{eqnarray}\label{eq:condition}
\forall \theta\in \Theta,\quad \prob_{\theta_i=0}(T_i = 1| D(i-1)) \le \E(\alpha_i|D(i-1))\,,
\end{eqnarray}
with $T_i$ given by equation~\eqref{eq:Ti}.
Then rule~\eqref{eq:RuleII} controls $\FDR$ and $\mFDR$ at level less than or equal to $\alpha$, i.e., for all $n \ge 1$, $\FDR(n) \le \alpha$ and $\mFDR(n) \le \alpha$.
%
%

%
\end{thm}
%
\begin{coro}
If the $p$-values are independent, then rule~\eqref{eq:RuleII} controls $\FDR$ and $\mFDR$ at level less than or equal to $\alpha$.
\end{coro}
Note that we do not require $p$-values to be independent, although it gives the simplest case where Condition~\eqref{eq:condition} holds true. Moreover, in~\eqref{eq:condition}, we condition on the number of discoveries so far. This is in contrast to alpha-investing method~\cite{alpha-investing} that uses information on acceptance of the previous hypotheses $(T_1, \dotsc, T_{i-1})$ or adaptive testing
 in a batch sequential arrivals (see e.g.~\cite{lehmacher1999adaptive}) that exploits information on the observed $z$-statistics.  
%
%
\begin{remark}
The update rule $\alpha_i = \beta_i (D(i-1)\vee 1)$ also controls $\FDR$, but leads to a slightly smaller number of 
discoveries. For this rule, the following holds true which is a more stringent control over $\mFDR$.
\begin{eqnarray}
\sup_{\theta\in \Theta} \, \frac{\E_\theta(V^\theta(n))}{\E_\theta(D(n) \vee 1)}\le \alpha\,.
\end{eqnarray}
\end{remark}

\subsection{{\sc \LORD} algorithm}
In the second algorithm the significance levels $\alpha_i$ are adapted to the time of last discovery, rather than
the number of discoveries so far. Concretely, choose any sequence of nonnegative numbers $\underline{\beta} = (\beta_i)_{i=1}^\infty$, such that $\sum_{i=1}^\infty \beta_i = \alpha$.
We then set
  $\alpha_i = \beta_i$ until a
discovery occurs. If $H_j$ is rejected, then the sequence is renewed by choosing $\alpha_{j+k} = \beta_k$, for $k\ge 1$, until we reach the next discovery. In other words, letting $\tau_i$ be the index of the most recent discovery before time step $i$,
$$\tau_i \equiv \max \Big\{\ell <i ,\, H_\ell \text{ is rejected} \Big\}\,,$$
(with $\tau_1 =0$) we set 
\begin{eqnarray}\label{eq:Rule}
\alpha_i = \beta_{i-\tau_i}\,.
\end{eqnarray}
In the next theorem we show that $\LORD$ controls $\mFDR_1(n)$ for all $n\ge 1$ and also controls $\FDR$ at every discovery.
\begin{thm}\label{thm:FDR}
Suppose that conditional on $\tau_{i-1}$, we have
\begin{eqnarray}\label{eq:condition1}
\forall \theta\in \Theta,\quad \prob_{\theta_i = 0}(T_i =1 \vert \tau_{i-1} )\le \E(\alpha_i|\tau_{i-1})\,.
\end{eqnarray}
Then, the rule~\eqref{eq:Rule} controls $\mFDR$ to be less than or equal to $\alpha$, i.e., $\mFDR_1(n) \le \alpha$ for all $n \ge 1$.
Further, it controls $\FDR$ at every discovery. More specifically, letting $\tau_k$ be the time of $k$-th discovery, we have the following for all $k \ge 1$, 
\begin{eqnarray}
\sup_{\theta\in \Theta} \E_\theta(\FDP(\tau_k) \ind(\tau_k < \infty) ) \le \alpha\,.\label{eq:FDR}
\end{eqnarray}
\end{thm}
\begin{coro}
If the $p$-values are independent, then rule~\eqref{eq:Rule} controls $\FDR$ and $\mFDR$ at level less than or equal to $\alpha$.
\end{coro}
\begin{remark}\label{rem:mFDR_strong}
For the update rule~\eqref{eq:Rule}, the following holds true which is stronger than the control of $\mFDR_1(n)$:
\begin{eqnarray}
\sup_{\theta\in \Theta} \, \frac{\E_\theta(V^\theta(n))}{\E_\theta(D(n-1) + 1)}\le \alpha\,.
\end{eqnarray}
\end{remark}
We refer to Section~\ref{proof-thm:FDR} for the proof of Theorem~\ref{thm:FDR} and Remark~\ref{rem:mFDR_strong}.
\subsection{Online control of FDR under dependency}\label{sec:dependence}
The BH procedure described in the introduction controls $\FDR$ when the $p$-values are independent. In~\cite{benjamini2001control}, Benjamini
and Yekutieli introduced a property called \emph{positive regression dependency from a subset} $I_0$ (PRDS on $I_0$) to capture the positive dependency structure among the test statistics. They relaxed the independence assumption on $p$-values by showing that if the joint distribution of the test statistics is PRDS on the subset of test statistics corresponding to true null hypotheses, then BH controls $\FDR$. (See Theorem 1.3 in~\cite{benjamini2001control}.) Further, they proved that BH controls $\FDR$ under general dependency if its threshold is adjusted by replacing $\alpha$ with $\alpha/(\sum_{i=1}^m \frac{1}{i})$ in equation~\eqref{eq:iBH}. Here, we prove an analogous result for \LOND algorithm for online controlling of $\FDR$.

\begin{thm}\label{thm:dependency-FDR}
If $\LOND$ is conducted with $\tilde{\beta}_i = \beta_i/(\sum_{j=1}^i \frac{1}{j})$ in place of $\beta_i$ in~\eqref{eq:RuleII}, it always controls $\FDR(n)$ at level less than or equal to $\alpha$, for all $n\ge 1$, without requiring condition~\eqref{eq:condition}.
\end{thm}

Theorem~\ref{thm:dependency-FDR} is proved in Section~\ref{proof:dependency-FDR}.

\section{Comparison with alpha-investing}
\label{sec:comparison-alpha}
Alpha-investing was developed by Foster and Stine~\cite{alpha-investing} to control $\mFDR$ in an online multiple testing problem. The basic
idea is to treat the significance level $\alpha$ as a budget to be spent over the sequence of tests and the rule earns an increment
in its budget each time it rejects a hypothesis. More precisely, alpha-investing rules assume an initial budget
$W(0)$, and a rule for significance levels $\alpha_i$ as a function of the form
\begin{eqnarray}\label{eq:alpha-investing}
\alpha_j = f_{W(0)}(\{T_1, T_2, \dotsc, T_{i-1}\})\,.
\end{eqnarray}
Let $W(k)\ge 0$ denote the budget after $k$ tests. The outcomes of the tests change the available budget as follows:
\begin{eqnarray*}
W(j) - W(j-1) = \begin{cases}
\omega & \text { if } T_j = 1\,\\
-\alpha_j/(1-\alpha_j) & \text { if } T_j = 0\,.
\end{cases}
\end{eqnarray*}
Choosing $\omega = \alpha$ and $W(0)\le \alpha \eta$, 
alpha-investing rules control $\mFDR_\eta$ under the condition that
the budget stays nonnegative almost surely~\cite{alpha-investing}. 
Note that since alpha-investing proceeds sequentially, it might stop
the testing after some number of rejected hypotheses. 

It is worth noting that \LOND and \LORD are not $\alpha$-investing rules. To show this, consider the case that all $p$-values are
equal to one. \LOND and \LORD do not reject any of the hypotheses in this scenario. Hence, both of them set the significance levels  $\alpha_j = \beta_j$ for all $j$ (cf. rule~\eqref{eq:RuleII},~\eqref{eq:Rule}). 
The budget after $n$ iteration works out at
$$W(n) = W(0) -\sum_{j=1}^n \frac{\beta_j}{1-\beta_j}\,.$$
Given that $W(0) \le \alpha$ and $\sum_{j=1}^\infty \beta_j = \alpha$, the above budget can become negative for some value of $n$. This is not allowed in the alpha-investing method though.

There is no general guarantee that alpha-investing controls FDR. Theorem 2 in ~\cite{alpha-investing} discusses a special case that
the testing procedure stops after a deterministic number of rejections and assumes that the procedure has uniform control of $\mFDR_0$.
In addition, it is proved to control $\mFDR$ under the assumption that $\prob_{\theta_i=0}(T_i=1|T_1,\dotsc,T_{i-1})\le \alpha_i$ for $i\ge 1$.
In contrast, \LOND and \LORD procedures control $\FDR$ and $\mFDR$. In Section~\ref{sec:dependence} we further proposed an adjustment to \LOND that addresses arbitrary dependency among the test statistics.
Section~\ref{sec:discovery-rate} studies the rate of expected number
of discoveries made by our procedures. We are not aware of any
analogous analysis for alpha-investing rules.

Let us finally mention that the recent work of  Aharoni and Rosset \cite{generalized-alpha} introduce a very broad
class of online rules called \emph{generalized $\alpha$-investing}. We
believe that our \LOND algorithm fits within this framework. Note however that
\cite{generalized-alpha}  --again-- does not guarantee $\FDR$ control for generalized
$\alpha$-investing. (Instead it proves $\mFDR$ control, adapting the
arguments of \cite{alpha-investing}.)

\section{Discussion}
\subsection{Effect of domain knowledge}\label{sec:domain-knowledge}
\LOND algorithm sets the significance levels based on the number of discoveries made so far.
Therefore there is an inherent positive effect from previous discoveries onto the next significance levels.
In other words, a larger number of discoveries leads to larger significance levels for future tests and hence more discoveries.
This property of \LOND is particularly beneficial when the researcher has knowledge about the underlying domain.
In this case, she can order hypotheses such that those that are most likely to be rejected (i.e., truly non-null hypotheses) appear first. This yields a large number of discoveries at the very first steps whose effect is everlasting on the
next significance levels. 
%

Likewise, \LORD can also leverage from the domain knowledge. It renews the sequence of significance  levels after each discovery.
Hence, if the truly non-null hypotheses arrive in batches, \LORD assigns a high significance level to them, namely $\beta_1$\footnote{Except for possibly the first hypothesis in the batch.}, which yields an increased statistical power. 
 
In Section~\ref{sec:numerical}, we study the effect of domain knowledge via numerical simulations.

\subsection{Choosing sequence $\underline{\beta}$}
Both \LOND and \LORD algorithms involve a sequence $\underline{\beta} = (\beta_\ell)_{\ell=1}^\infty$ such that
$\sum_{\ell=1}^\infty \beta_\ell = \alpha$. Examples of such sequence are:
\begin{eqnarray*}
\beta_\ell &=& \frac{C(\alpha,\nu)}{\ell^a}\,,\\
\beta_\ell &=& \frac{\tC(\alpha,\nu)}{\ell \log^\nu(\ell \vee 1)}\,,
\end{eqnarray*}
where $a, \nu >1$ and constants $C, \tC$ are chosen in a way to ensure $\sum_{\ell=1}^\infty \beta_\ell =\alpha$.

In case there is an upper bound $n$ on the number of hypotheses to be tested, such information can be exploited by choosing
$C$ or $\tC$ to satisfy $\sum_{\ell=1}^n \beta_\ell =\alpha$. This leads to larger values of $\beta_\ell$ and in turn larger significance
levels for tests. Further, parameter $\nu$ controls how fast the sequence $\underline{\beta}$ decreases. If there is prior information about the arrival times of truly false hypotheses (e.g., batch patterns, size of the batches, ...), this information can be used in choosing proper value of $\nu$. 
\section{Total discovery rate}\label{sec:discovery-rate}
The proposed procedures control FDR and mFDR to be below $\alpha$. 
However, apart from FDR the total discovery rate is another important characteristic of 
a testing procedure. For instance, the procedure that never rejects a hypothesis achieves zero FDR but is useless.
 
In this section, we aim at comparing the total discovery rate for BH and our algorithms. We show that although our algorithms use only the outcomes of the previous tests in the stream, they have comparable performance to BH in terms of discovery rate.

We focus on the mixture model wherein each null hypothesis $H_i$ is false with probability $\ve$ independently of other hypotheses. Further, the test statistics and hence the $p$-values are independent. Under null hypothesis $H_i$ we have $p_i \sim \cU([0,1])$ and under its alternative, $p_i $ is distributed according to a non-uniform distribution whose CDF is denoted by $F$.  Clearly, in this model the $p$-values have the same marginal distribution. 
Let $G(\cdot)$ be the CDF of this marginal distribution, i.e., $G(x) = (1-\ve) x + \ve F(x)$. For clarity in presentation, we assume $F(x)$, and hence $G(x)$, is continuous.

%
%

\subsection{Discovery rate of {\sc \LOND}}
In the next theorem we lower bound total discovery rate of \LOND for a specific choice of sequence $\underline{\beta} = (\beta_i)_{i=1}^\infty$. The goal here is to show that \LOND has nearly linear number of discoveries, namely $\E(D^\LOND(n)) = \Omega(n^{1-a})$, for arbitrary $0<a<1$ and all $n\ge1$. 
%
\begin{thm}\label{thm:LOND_rate}
Consider the mixture model for the $p$-values with $G(x) = (1-\eps) x+\eps F(x)$ denoting the CDF of the marginal distribution.
Assume that for some constants $x_0$, $\lambda >0$ and $\kappa\in (0,1)$, we have $F(x) \ge \lambda x^\kappa$, for all $x < x_0$.  
Choose sequence $\underline(\beta) = (\beta_i)_{i=1}^\infty$ as $\beta_i = Ci^{-\nu}$ with $1<\nu<1/\kappa$ and $C$ set such that $\sum_{i=1}^\infty \beta_i = \alpha$. We further let $D^{\LOND}(n)$ denote the number of discoveries by applying \LOND algorithm with sequence $\underline{\beta}$ to set of $p$-values $(p_1, p_2, \dotsc, p_n)$.  Then, there exists constant $\widetilde{C}$ such that
for any fixed $0<\delta<1$ and any $n \ge 1$, the following holds.
\begin{align}\label{eq:LOND_rate}
\prob\Big\{ D^{\LOND}(n) \ge \Big(\frac{\delta n}{\tC}\Big)^{\frac{1-\kappa\nu}{1-\kappa}}\Big\} \ge 1-\delta\,.
\end{align}
\end{thm}
We refer to Section~\ref{proof:LOND_rate} for the proof of Theorem~\ref{thm:LOND_rate}. 

Equation\eqref{eq:LOND_rate} implies that
\begin{align*}
\E(D^\LOND(n)) \ge (1-\delta) \Big(\frac{\delta n}{\tC}\Big)^{\frac{1-\kappa\nu}{1-\kappa}}\,.
\end{align*}
Note that since $1<\nu<1/\kappa$ is arbitrary, the exponent $(1-\kappa\nu)/(1-\kappa)$ can be made arbitrarily close to 1. 
 It is therefore clear form equation~\eqref{eq:LOND_rate} that for any fixed $0<a<1$ and all $n\ge 1$, we have
 $$\E(D^\LOND(n)) = \Omega(n^{1-a})\,.$$

%

\subsection{Discovery rate of {\sc \LORD}}
As we show in the following theorem, \LORD algorithm leads to a linear total discovery rate.
\begin{thm}\label{thm:EDm}
Consider the mixture model for the $p$-values with $G(x) = (1-\eps) x+\eps F(x)$ denoting the CDF of the marginal distribution.
We further let $D^{\LORD}(n)$ denote the number of discoveries by applying \LORD algorithm with sequence $\underline{\beta}$ to set of $p$-values $\{p_1, p_2, \dotsc, p_n\}$. Then, $\lim_{n\to \infty} ( D^\LOND(n)/n)$ exists almost surely and 
\begin{eqnarray}
\lim_{n\to \infty} \frac{1}{n}\, D^\LOND(n) \ge \cA(G,\underline{\beta})\,,
\quad \quad \cA(G,\underline{\beta}) \equiv \Big(\sum_{k=1}^\infty e^{-\sum_{\ell=1}^k G(\beta_\ell)} \Big)^{-1}\,.
\end{eqnarray}
Further,
\begin{eqnarray}\label{eq:LORD_E}
\lim_{n\to \infty} \frac{1}{n}\, \E(D^\LOND(n)) \ge \cA(G,\underline{\beta})\,.
\end{eqnarray}
\end{thm}

\begin{coro}\label{cor:EDm}
Suppose that there exists a sequence $\underline{\beta} = (\beta_\ell)_{\ell=1}^\infty$ such that the following conditions hold true:
\begin{enumerate}
\item $\forall \ell,\, \beta_\ell \in (0,1)$ and $\sum _{\ell=1}^\infty \beta_\ell =\alpha\,.$
\item There exist constants $L_0>0$ and $c >1$ such that for $\ell>L$, we have $G(\beta_\ell) > c/\ell$.
\end{enumerate}
Then, we have almost surely
\begin{align}
\lim_{n\to \infty} \frac{1}{n} D^\LOND(n) \ge \Big(L_0+ \frac{1}{(c-1)L_0^{c-1}}\Big)^{-1}\,.
\end{align}
Similarly,
\begin{align}
\lim_{n\to \infty} \frac{1}{n} \E(D^\LOND(n)) \ge \Big(L_0+ \frac{1}{(c-1)L_0^{c-1}}\Big)^{-1}\,.
\end{align}
\end{coro}

In words, upon finding sequence $\underline{\beta}$ which satisfies conditions of Corollary~\ref{cor:EDm}, \LORD gives linear discovery rate $$D^\LORD (n) = \Theta(n)\,.$$

\bigskip

{\bf Example 1.} Suppose that there exist constants $\lambda > 0$, $\kappa \in (0,1)$, and $x_0$ such that $F(x) \ge \lambda x^\kappa$, for all $x < x_0$. Define
$$C_* \equiv \sum_{\ell=1}^\infty \frac{\log \ell}{\ell^{1/\kappa}}<\infty\,.$$ 
Letting $\beta_\ell =  (\alpha \log \ell)/(C_* \ell^{1/\kappa})$, we have $\sum_{\ell=1}^\infty \beta_\ell = \alpha$.  Also, for any $c>1$ we can choose $L_0$ sufficiently large such that for $\ell >L_0$
\begin{align*}
G(\beta_\ell) &\ge (1-\eps) \beta_\ell + \eps \lambda \beta_\ell^\kappa\\
&\ge \eps \lambda \frac{(\alpha \log \ell)^\kappa}{C_*^\kappa \ell} > \frac{c}{\ell}\,.
\end{align*}
Therefore, both conditions in Corollary~\ref{cor:EDm} are satisfied and thus \LORD gives linear discovery rate.
\bigskip

{\bf Example 2 (Mixture of Gaussians).} Suppose that we are getting samples $z_i \sim \normal(\theta_i,1)$ and we want to test hypothesis $H_i :\, \theta_i = 0$ versus the alternative $\theta_i = \mu$. In this case, two-sided $p$-values are given by 
$$p_i = 2\Big(1-\Phi(|z_i|)\Big)\,.$$
Therefore,
\begin{eqnarray*}
F(\nu) &=& \prob_{\theta_i=\mu} (p_i\le \nu) \\
&=& \prob_{\theta_i=\mu}(\Phi^{-1}(1-\nu/2) \le |z_i|)\\
&=& 2- \Phi(\Phi^{-1}(1-\nu/2) + \mu) - \Phi(\Phi^{-1}(1-\nu/2)-\mu)\,.
\end{eqnarray*}
%
Let $\zeta \equiv \Phi^{-1}(1-\nu/2)$ and thus $\Phi(-\zeta) = \nu/2$. Using this notation, we write
\begin{eqnarray}
F(\nu) = 2- \Phi(\zeta+\mu) - \Phi(\zeta-\mu)
> {\Phi(\mu-\zeta)}\,. \label{eq:Mix1}
\end{eqnarray}
Recall the following classical bounds on the CDF of normal distribution for $t\ge0$: 
\begin{align}
\frac{\phi(t)}{t}\Big(1-\frac{1}{t^2}\Big)\le \Phi(-t) \le \frac{\phi(t)}{t}\,.\label{eq:normal_bound}
\end{align}
Applying the above inequalities in equation~\eqref{eq:Mix1}, we obtain
\begin{align}
F(\nu) &> \Phi(-\zeta) \frac{\Phi(\mu-\zeta)}{\Phi(-\zeta)} \nonumber\\
&\ge \frac{\nu}{2} \frac{\phi(\zeta-\mu)}{\phi(\zeta)} \frac{\zeta}{\zeta-\mu} \Big(1-\frac{1}{(\zeta-\mu)^2}\Big)\nonumber\\
&= \frac{\nu}{2}\, e^{\mu \zeta - \mu^2/2}\,\frac{\zeta}{\zeta-\mu} \Big(1-\frac{1}{(\zeta-\mu)^2}\Big)\nonumber\,.
\end{align}
Using the LHS bound in~\eqref{eq:normal_bound}, it is easy to see that for small enough $\nu$, $\zeta> \sqrt{\log (2/\nu)}$. Therefore, for small enough $\nu$, we obtain
\begin{eqnarray}
F(\nu) > \frac{\nu}{4} \exp(-\mu^2/2) \exp(\mu\sqrt{\log(2/\nu)})\,.\label{eq:Mix2}
\end{eqnarray}
Fix $a>1$ and let 
$$C_* = \sum_{\ell=1}^\infty \frac{1}{\ell \log^a(\ell+1)} < \infty\,.$$
Choosing 
$$\beta_\ell =  \frac{\alpha}{C_* \ell \log^a (\ell+1)}\,,$$
clearly $\sum_{\ell=1}^\infty \beta_\ell = \alpha$. Furthermore, invoking equation~\eqref{eq:Mix2} it is easy to see that for any $c>1$ there exists $L_0$ sufficiently large,
such that the following holds true for all $\ell >L_0$:
\begin{eqnarray*}
G(\beta_\ell) \ge \eps F(\beta_\ell) \ge \frac{c}{\ell}\,.
\end{eqnarray*}
Therefore, both conditions in Corollary~\ref{cor:EDm} are satisfied and the expected number of discoveries would be $\Theta(n)$.

\section{Numerical experiments}\label{sec:numerical}
In this section, we compare the performance of \LOND and \LORD algorithm with BH, Bonferroni and alpha-investing procedures in terms of $\FDR$,
$\mFDR$ and statistical power.  Our evaluations are on both synthetic and real data sets.
\subsection{Synthetic data}
\subsubsection{Independent $p$-values}\label{sec:numerical-indep}
We consider similar setup as in~\cite{alpha-investing}. A set of hypotheses $\cH(n) = (H_1,\cdots, H_n)$ are tested where each hypothesis concerns mean of a normal distribution, $H_j : \,\theta_j = 0$.
Parameters $\theta_j$ are set according to a mixture model:
\begin{eqnarray}\label{eq:mu}
\theta_j \sim \begin{cases}
0 & w.p. \quad 1-\pi\,,\\
\normal(0,\sigma^2) & w.p. \quad \pi\,.
\end{cases}
\end{eqnarray}

\bigskip

 In the first experiment, we set $n=1000$ and $\sigma^2 = 2\log n$. The test statistics are independent normal random variables $Z_j \sim \normal(\theta_j,1)$, for which the two-sided $p$-values work out at $p_j = 2\Phi(-|Z_j|)$.\footnote{Value of $\sigma$ controls the strength of the signal to be distinguished from noise. We set $\sigma^2 = 2\log n$ because for truly null hypotheses ($\theta_i = 0$) we have $Z_i\sim\normal(0,1)$, and maximum of $m$ normal variables  is w.h.p. at most $\sqrt{2\log(n)}$. Clearly, larger $\sigma$ leads to larger power and better control of FDR and mFDR.} Given these $p$-values, we compute $\FDR$ and $\mFDR_1$ for significance level $\alpha = 0.05$ by averaging over $10,000$ trials of the sequence of test statistics. 
For \LOND and \LORD we choose the sequence $\underline{\beta} = (\beta_\ell)_{\ell=1}^\infty$ as follows:
\begin{eqnarray}
\beta_\ell = \frac{C_*}{\ell \log^2(\ell \vee 2)}\,,
\end{eqnarray}
with $C_*$ set in a way to ensure $\sum_{\ell=1}^\infty \beta_\ell = \alpha$.

The Bonferroni procedure used in our experiments is the one that applies the following decision rule:
\begin{eqnarray*}
T_\ell = \begin{cases}
1 &\text{if }p_\ell \le \beta_\ell \quad \text{(reject the null hypothesis $H_\ell$)}\,,\\
0 &\text{otherwise} \quad \text{(accept the null hypothesis $H_\ell$)}.
\end{cases}
\end{eqnarray*} 

We consider an alpha-investing rule, as explained in Section~\ref{sec:comparison-alpha}, with the following rule in equation~\eqref{eq:alpha-investing}:
$$\alpha_j  = \frac{W(j)}{1+j-\tau_j}\,,$$
where $\tau_j$ denotes the time of the most recent discovery by time $j$.  This rule was proposed by~\cite{alpha-investing} to show how context dependent information can be incorporated into building alpha-investing rules. In case there is substantial side information that the first few hypotheses are likely to be rejected and that the truly non-null hypotheses appear in clusters, this rule exploits such information to increase power. The same rule has been used by~\cite{lin2011vif} in VIF regression algorithm designed for online feature selections.   

We evaluate performance of the five algorithms under the following two scenarios:
\begin{itemize}
\item[$\bullet$] \emph{Scenario I: Absence of domain knowledge.} In this scenario, nonzero means $\mu_i$ appear randomly in the stream of
hypotheses as explained by equation~\eqref{eq:mu}.

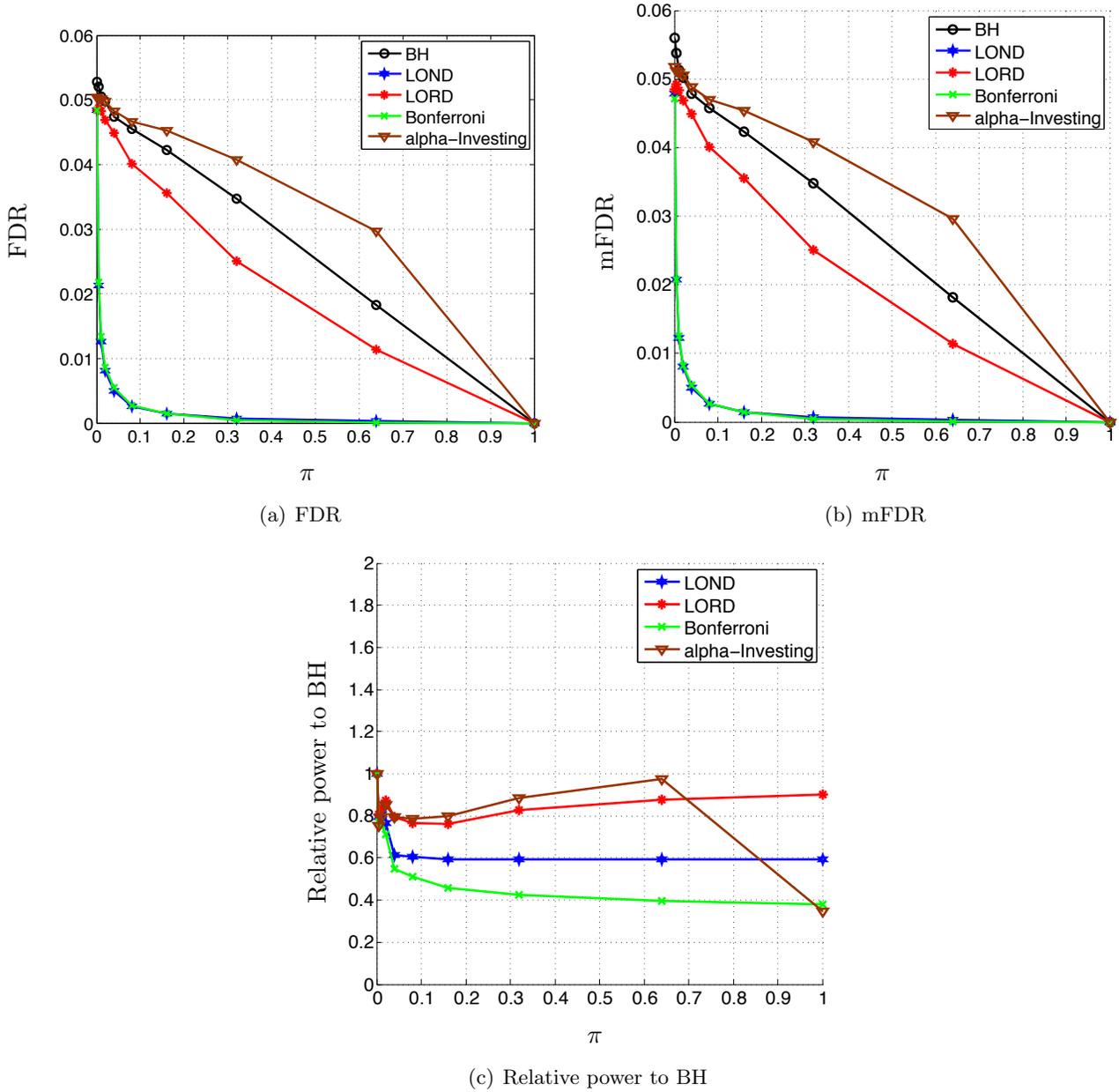
\begin{figure}[!t]
    \centering
    \subfigure[FDR]{
        \includegraphics[width = 2.9in]{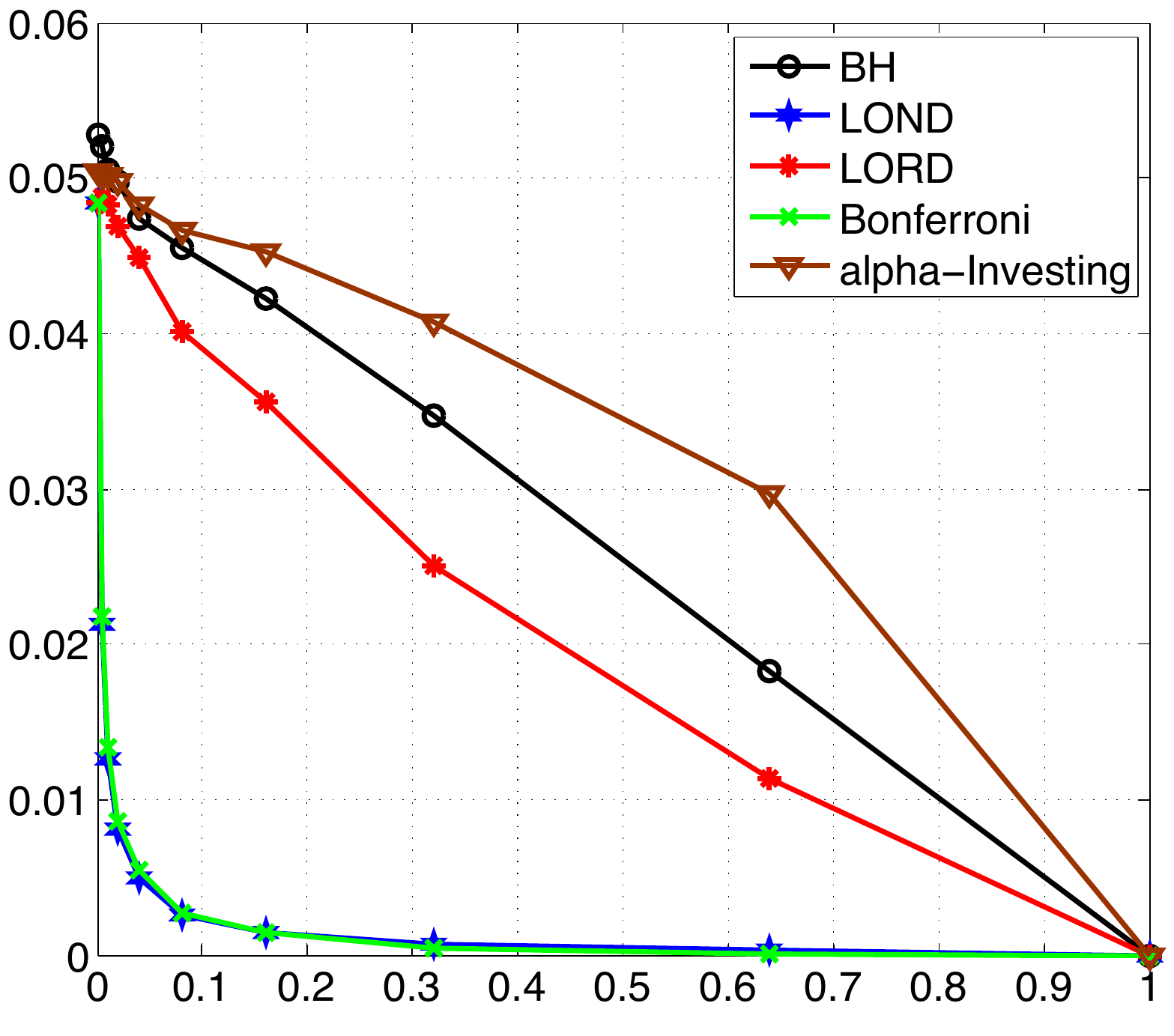}
        \put(-105,-10){$\pi$}
        \put(-105,-20){\phantom{a}}
        \put(-230,85){\rotatebox{90}{$\FDR$}}
        \label{fig:FDR_Worst}
        }
        \hspace{0.8cm}
    \subfigure[mFDR]{
        \includegraphics[width=2.9in]{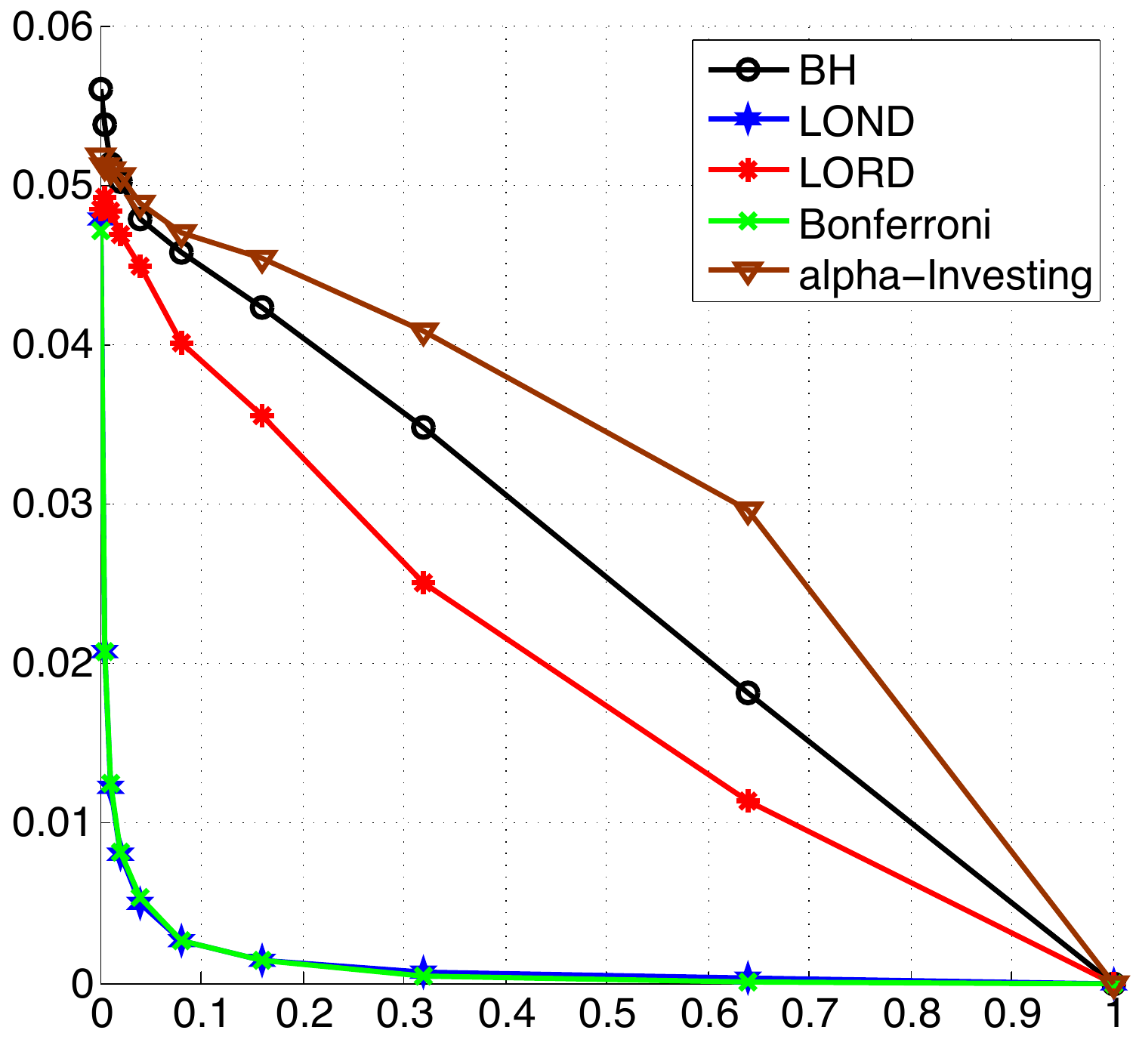}
        \put(-105,-10){$\pi$}
        \put(-105,-20){\phantom{a}}
        \put(-225,80){\rotatebox{90}{$\mFDR$}}
        \label{fig:mFDR_Worst}
        }
     \subfigure[Relative power to BH]{
        \includegraphics[width=2.9in]{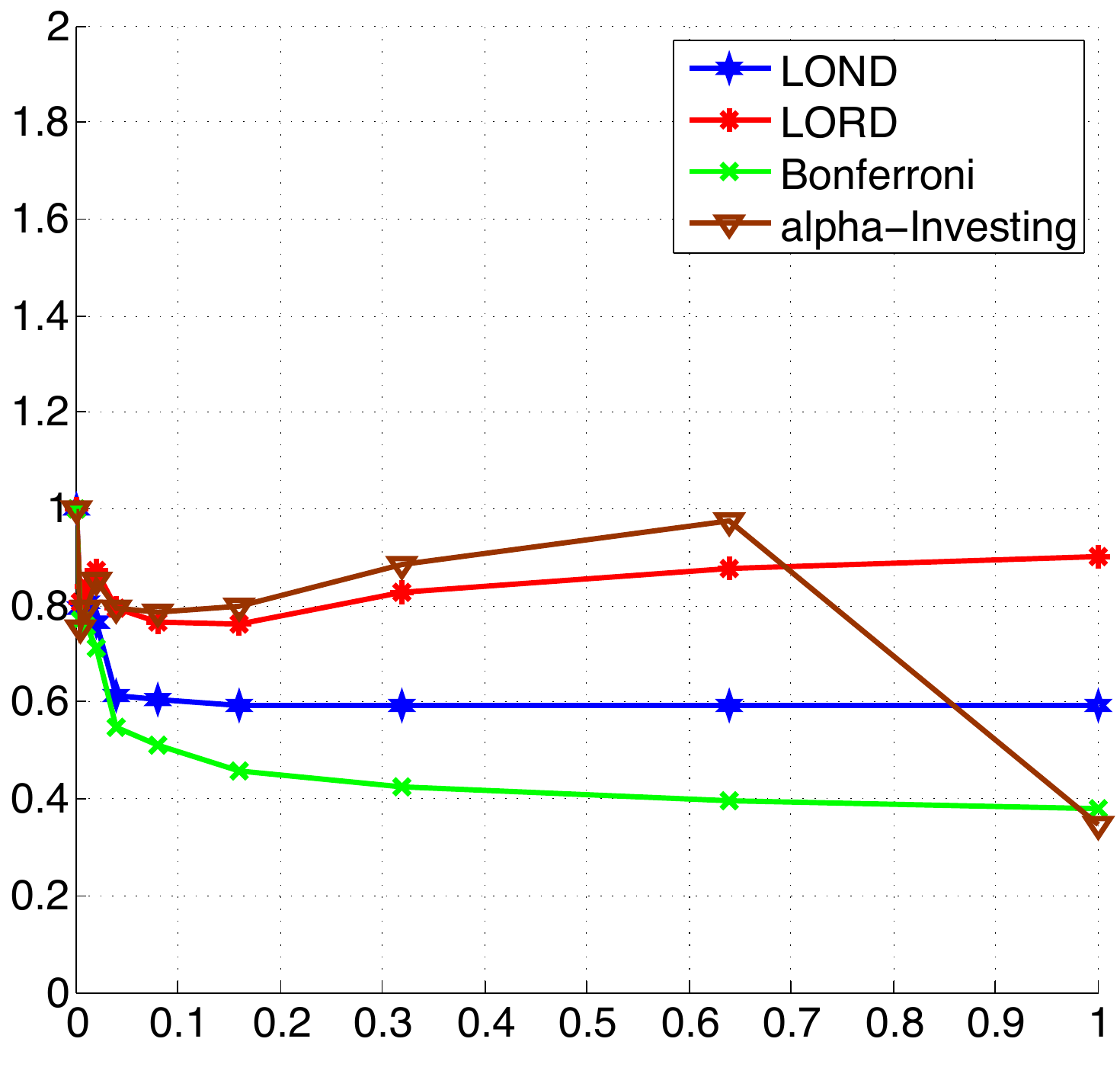}
        \put(-105,-10){$\pi$}
        \put(-105,-20){\phantom{a}}
        \put(-225,50){\rotatebox{90}{Relative power to BH}}
        \label{fig:Power_Worst}
        }
    \caption{FDR and mFDR of \LORD, \LOND, BH, Bonferroni and an alpha-investing rule under Scenario I. Figure (c) shows the relative power of the procedures to BH method.}\label{fig:Exp1_Worst}
    \vspace{-.7cm}
\end{figure}

\item[$\bullet$] \emph{Scenario II: Presence of domain knowledge.} We assume that an investigator has knowledge about the underlying domain of research and she uses this information in choosing the hypotheses. She first tests the (null) hypotheses that are most likely to be rejected. To simulate this scenario, we sort the hypotheses in the stream according to the absolute values of means, i.e., $|\theta_i|$, in decreasing order. Those with larger $|\theta_i|$ appear earlier. Let us stress that the ordering is based on $|\theta_j|$ and not the $p$-values.
\end{itemize}
\begin{figure}[!t]
    \centering
    \subfigure[FDR]{
        \includegraphics[width = 2.9in]{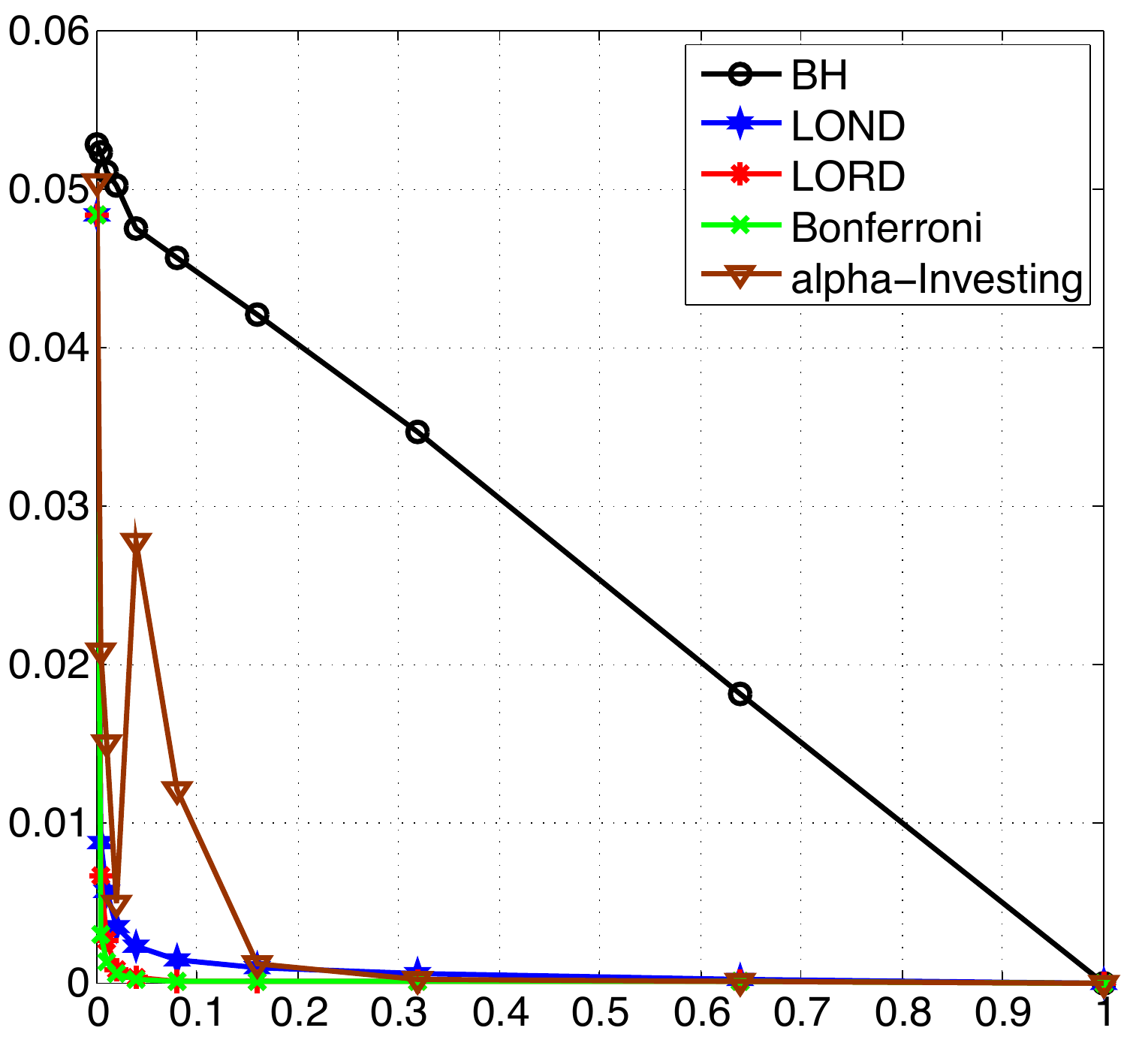}
        \put(-105,-10){$\pi$}
        \put(-105,-20){\phantom{a}}
        \put(-225,80){\rotatebox{90}{$\mFDR$}}
        \label{fig:FDR_Best}
        }
         \hspace{0.8cm}
    \subfigure[mFDR]{
        \includegraphics[width=2.9in]{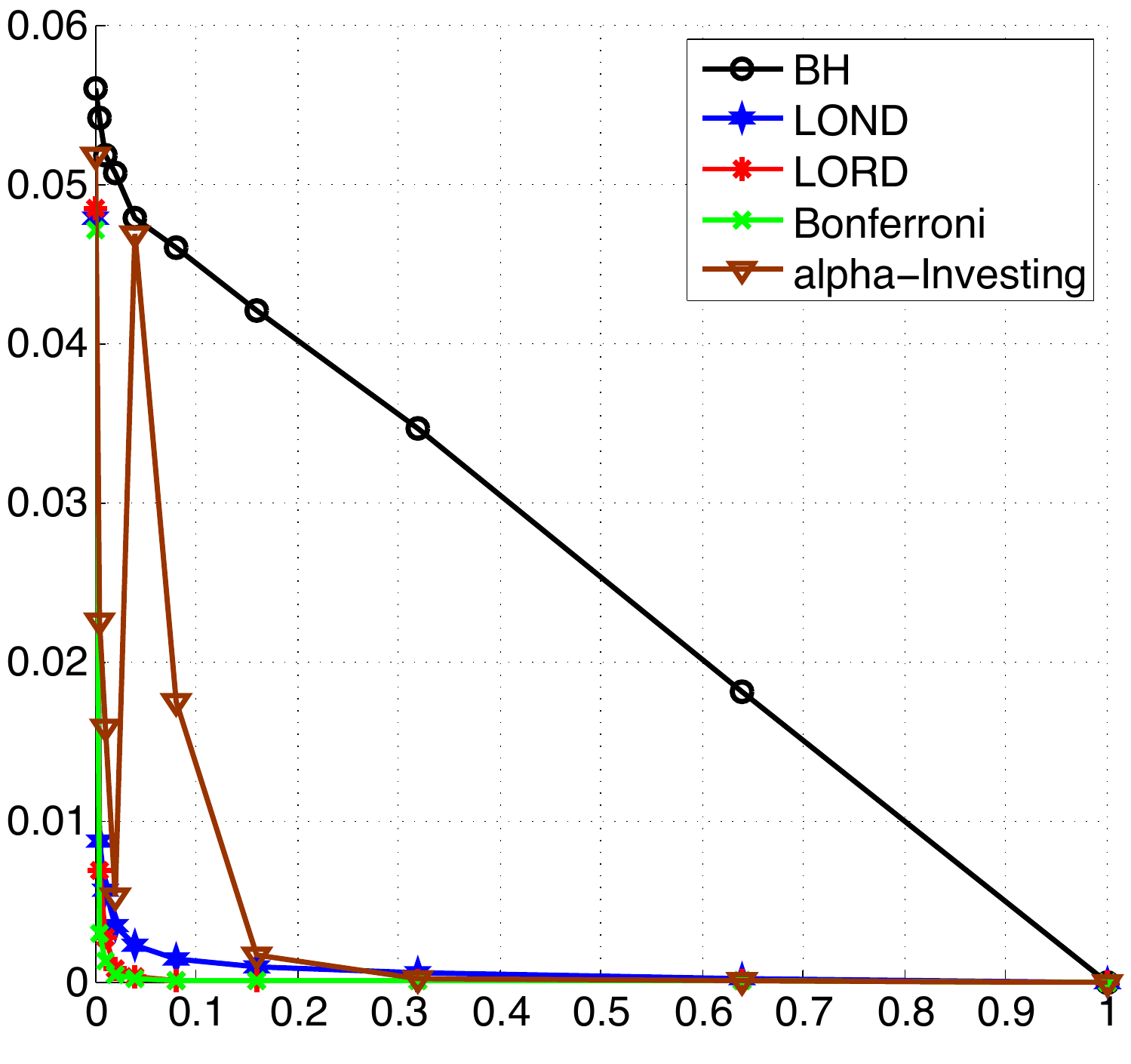}
        \put(-105,-10){$\pi$}
        \put(-105,-20){\phantom{a}}
        \put(-225,80){\rotatebox{90}{$\mFDR$}}
        \label{fig:mFDR_Best}
        }
     \subfigure[Relative power to BH]{
        \includegraphics[width=2.9in]{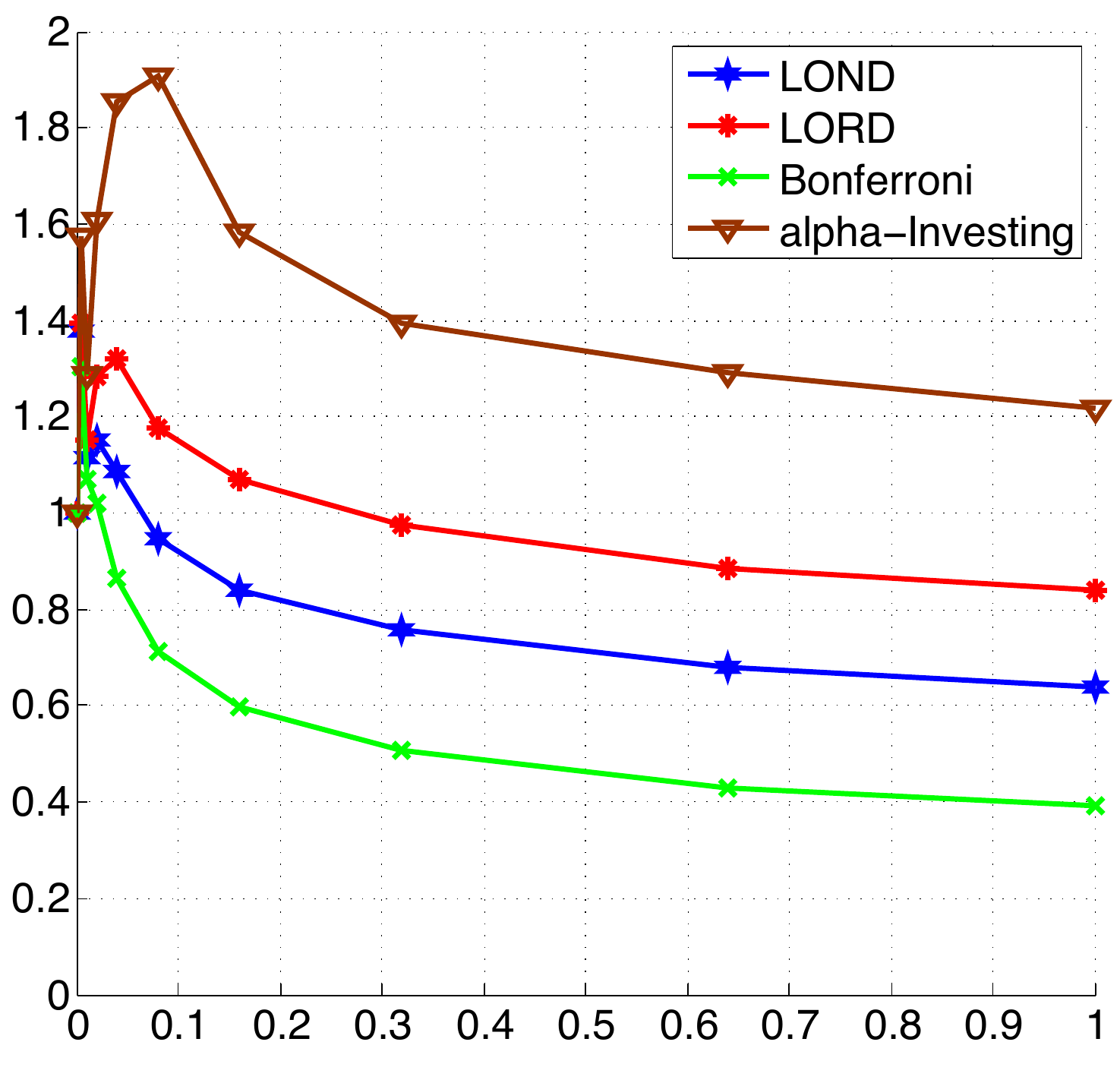}
        \put(-105,-10){$\pi$}
        \put(-105,-20){\phantom{a}}
        \put(-225,50){\rotatebox{90}{Relative power to BH}}
        \label{fig:Power_Best}
        }
    \caption{FDR and mFDR of \LORD, \LOND, BH, Bonferroni and an alpha-investing rule under Scenario II. Figure (c) shows the relative power of the procedures to BH method.}\label{fig:Exp1_Best}
    \vspace{-.7cm}
\end{figure}

\begin{figure}[!t]
    \centering
    \subfigure[\LOND]{
        \includegraphics[width = 2.9in]{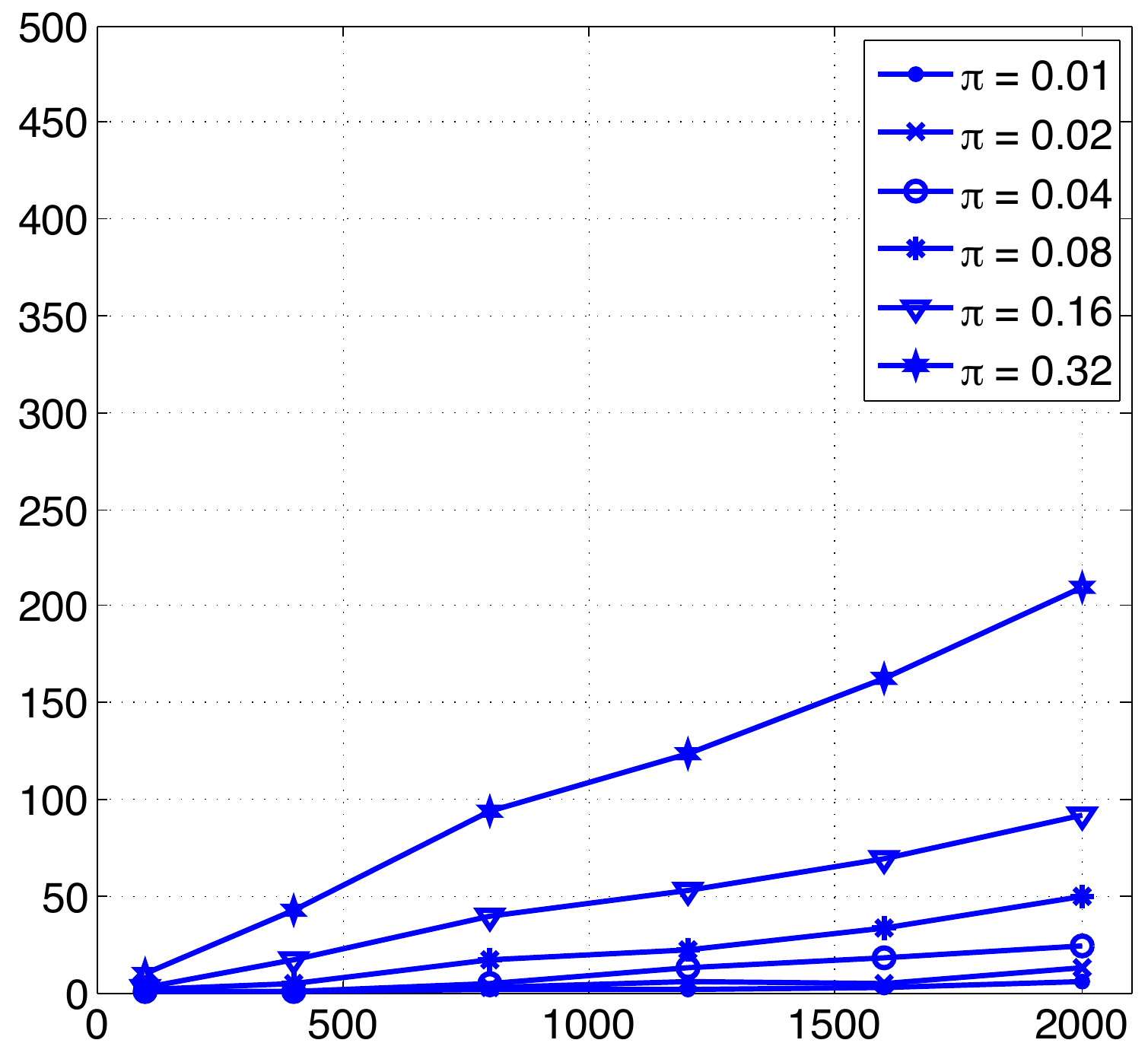}
        \put(-105,-10){$n$}
        \put(-105,-20){\phantom{a}}
        \put(-225,35){\rotatebox{90}{Total number of discoveries}}
        \label{fig:LOND_Worst}
        }
         \hspace{0.8cm}
    \subfigure[\LORD]{
        \includegraphics[width=2.9in]{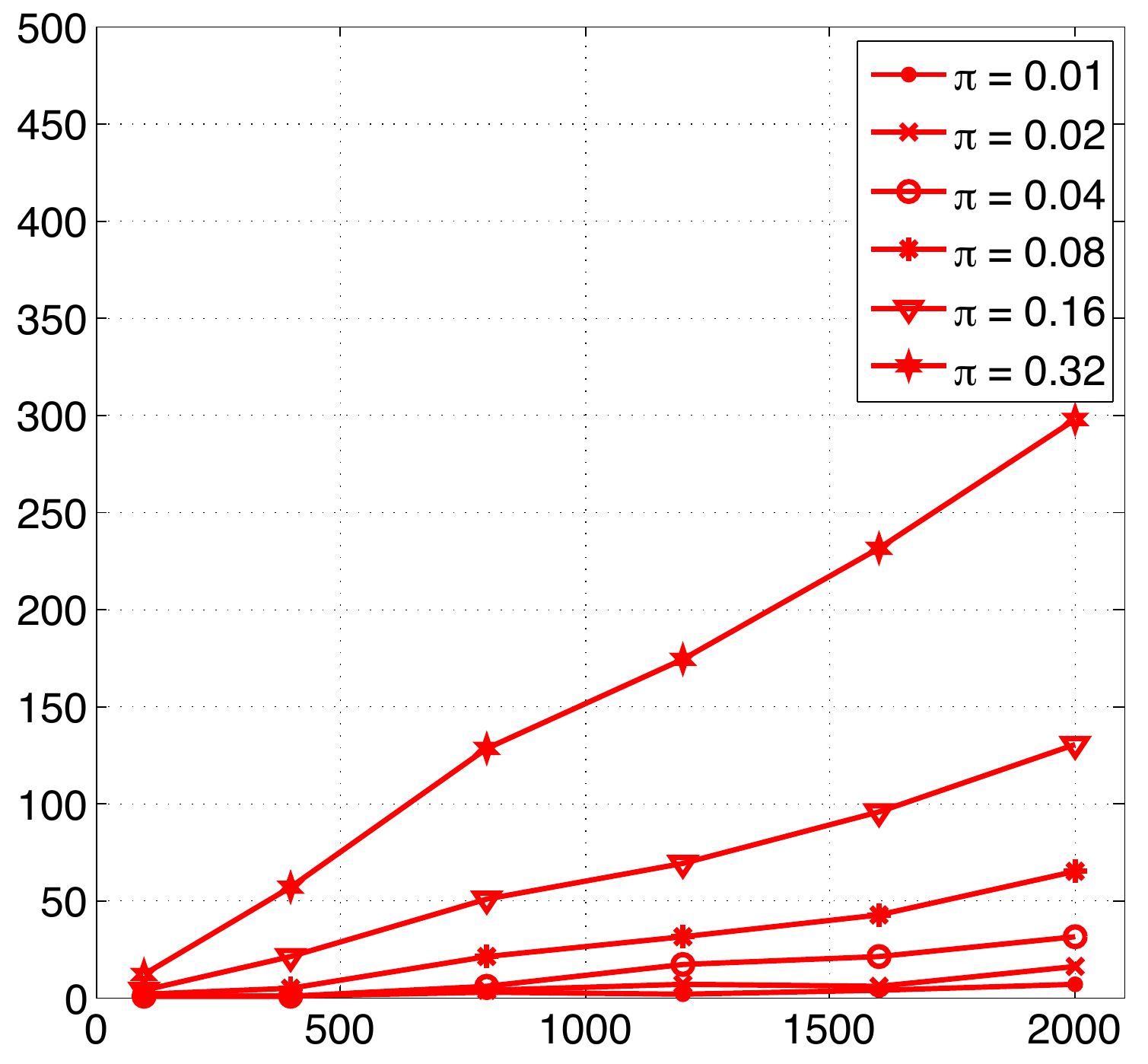}
        \put(-105,-10){$n$}
        \put(-105,-20){\phantom{a}}
        \put(-225,35){\rotatebox{90}{Total number of discoveries}}
        \label{fig:LORD_Worst}
        }
    \caption{Total number of discoveries under setting described in Section~\ref{sec:numerical-indep} (Scenario I). }\label{fig:totD_Worst}
    \vspace{-.7cm}
\end{figure}

\begin{figure}[!t]
    \centering
    \subfigure[\LOND]{
        \includegraphics[width = 2.9in]{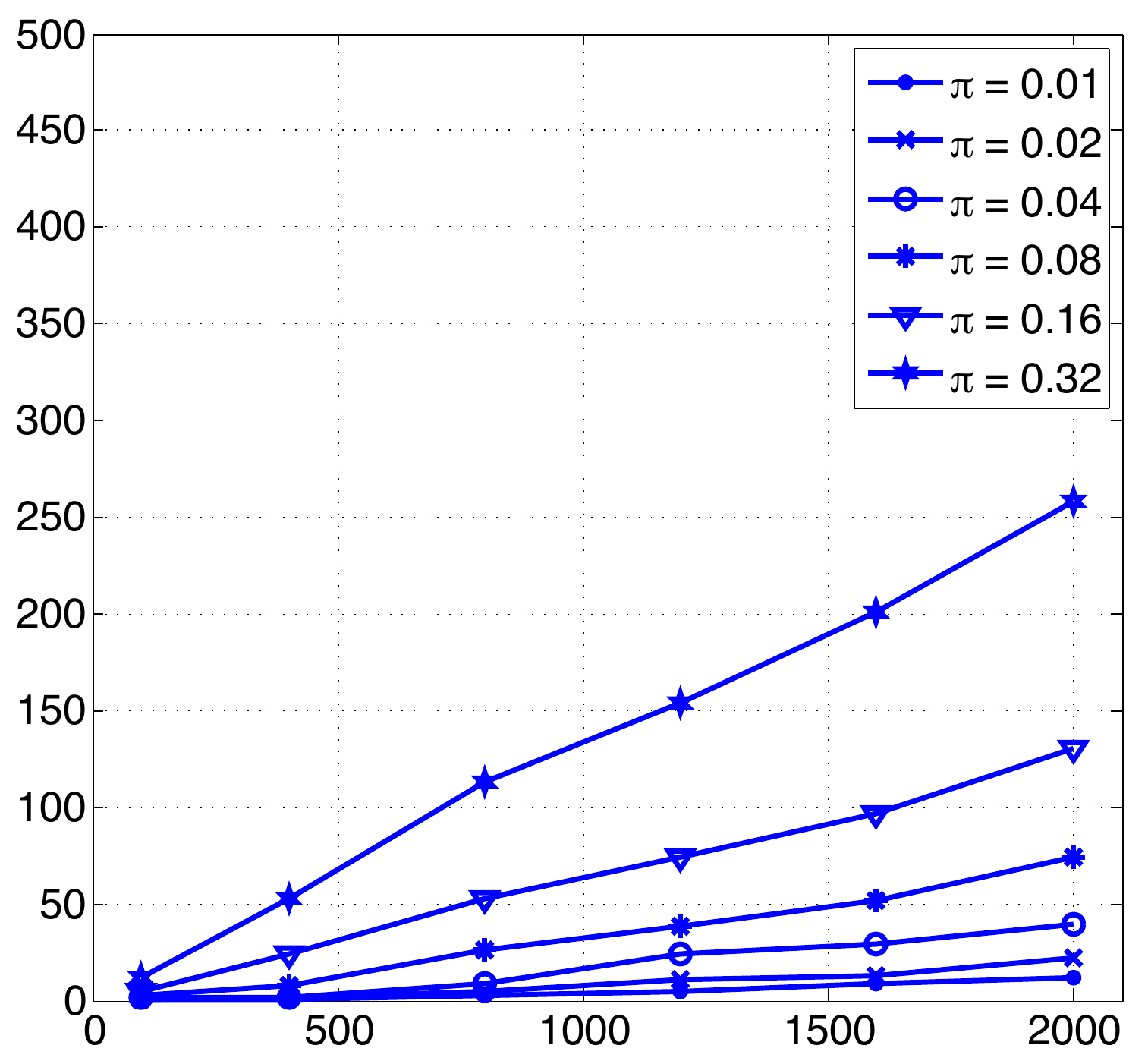}
        \put(-105,-10){$n$}
        \put(-105,-20){\phantom{a}}
        \put(-225,35){\rotatebox{90}{Total number of discoveries}}
        \label{fig:LOND_Best}
        }
         \hspace{0.8cm}
    \subfigure[\LORD]{
        \includegraphics[width=2.9in]{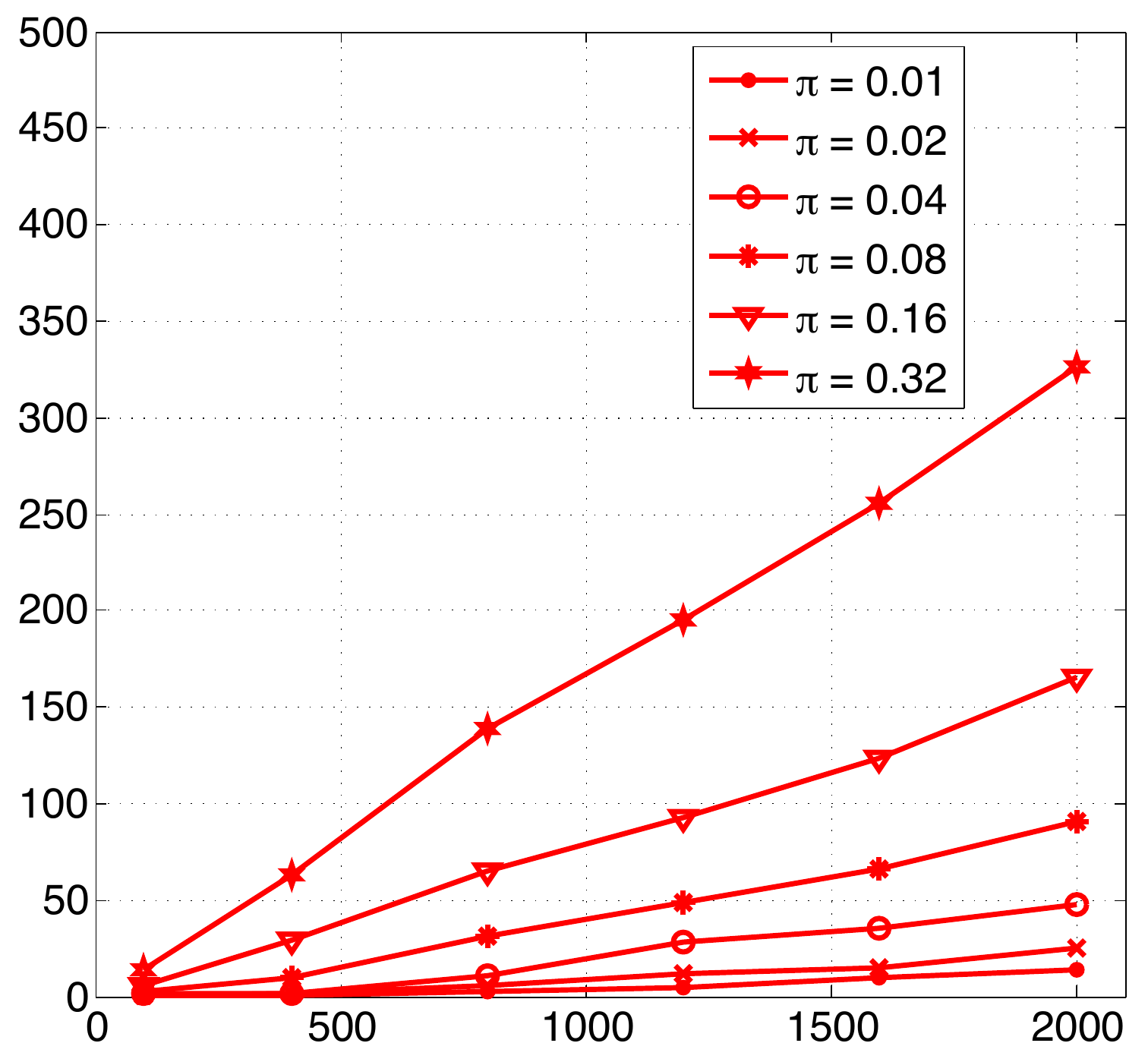}
        \put(-105,-10){$n$}
        \put(-105,-20){\phantom{a}}
        \put(-225,35){\rotatebox{90}{Total number of discoveries}}
        \label{fig:LORD_Best}
        }
    \caption{Total number of discoveries under setting described in Section~\ref{sec:numerical-indep} (Scenario II).}\label{fig:totD_Best}
    \vspace{-.7cm}
\end{figure}

Figures~\ref{fig:FDR_Worst} and~\ref{fig:mFDR_Worst} show FDR and mFDR achieved by the five procedures and several values of $\pi$, the proportion of truly non-null hypotheses, under Scenario I. As we see, FDR and mFDR of \LOND and Bonferroni are almost identical under this setting.
In Figure~\ref{fig:Power_Worst}, we show the relative power of the procedures with respect to BH method, under Scenario I. More precisely, letting $U_\theta(n) = D(n) - V^\theta(n)$ be the number of correctly rejected hypotheses by the procedure of interest, we estimate 
$$\E\left(\frac{U_\theta(n)}{U^\BH_\theta(n)}\right)$$
via averaging over $10,000$ trials of test statistics. Note that for large values of $\pi$, the relative power of alpha-investing drops substantially. The reason is that in this case, the rule rejects most of the hypotheses at the beginning and its budget, and ergo the significance level $\alpha_\ell$, increases rapidly. When $\alpha_\ell$ gets close to one, any acceptance of a null hypothesis yields a large decrease in the budget and makes it negative. Therefore, the algorithm halts henceforth missing the next discoveries.

Figure~\ref{fig:Exp1_Best} exhibits the same metrics for the five procedures under Scenario II. In presence of domain knowledge, we see faster drop-off in $\FDR$ and $\mFDR$ of \LORD and \LOND procedures. Further, they achieve higher relative power to BH compared to Scenario I, especially for small values of $\pi$. This supports our discussion in Section~\ref{sec:domain-knowledge}.

\bigskip

In the second experiment, we compute the expected number of discoveries made by \LORD and \LOND under Scenarios I and II, for several values of $n$ and $\pi$.  The results are depicted in Figures~\ref{fig:totD_Worst} and~\ref{fig:totD_Best}. As we see, for any fix $\pi$, \LORD and \LOND exhibit a linear discovery rate in both scenarios. This corroborates our findings in Section~\ref{sec:discovery-rate} regarding the total discovery rate of these procedures.

\begin{figure}[!t]
    \centering
    \subfigure[FDR]{
        \includegraphics[width = 2.9in]{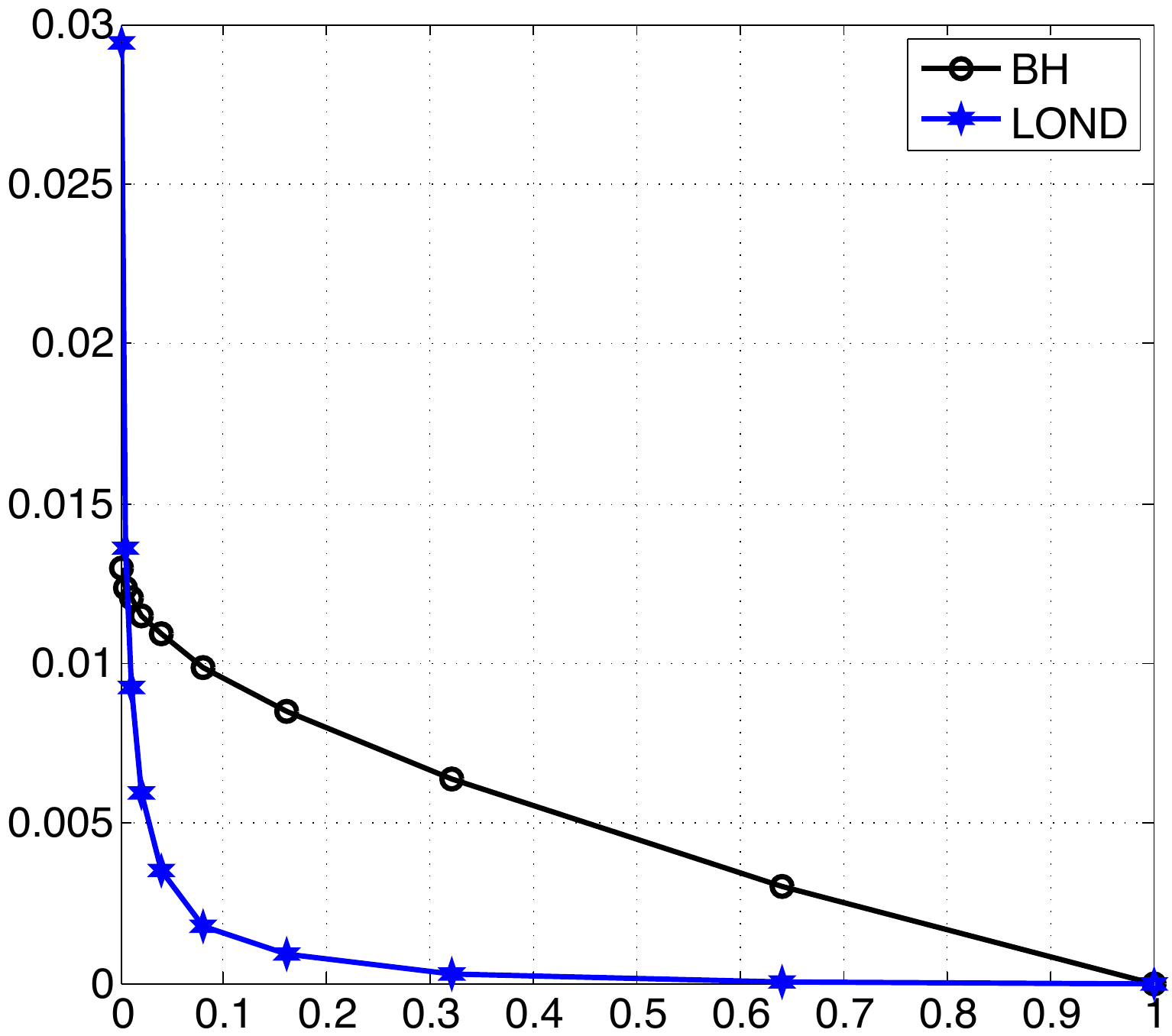}
        \put(-105,-10){$\pi$}
        \put(-105,-20){\phantom{a}}
        \put(-225,80){\rotatebox{90}{$\FDR$}}
        \label{fig:FDR_Worst_dep2}
        }
        \hspace{.8cm}
    \subfigure[mFDR]{
        \includegraphics[width=2.9in]{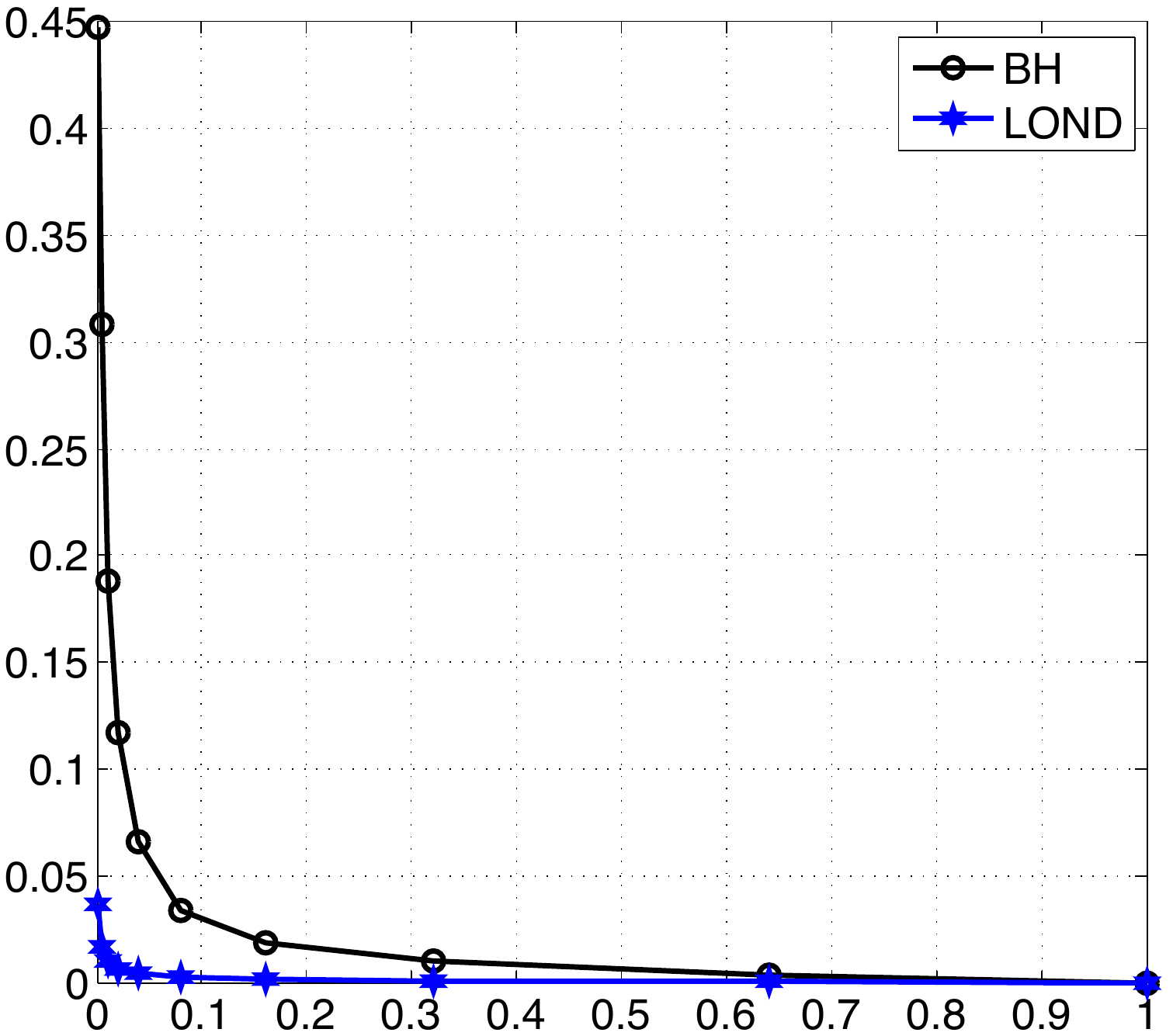}
        \put(-105,-10){$\pi$}
        \put(-105,-20){\phantom{a}}
        \put(-225,80){\rotatebox{90}{$\mFDR$}}
        \label{fig:mFDR_Worst_dep2}
        }
     \subfigure[Relative power to BH]{
        \includegraphics[width=2.9in]{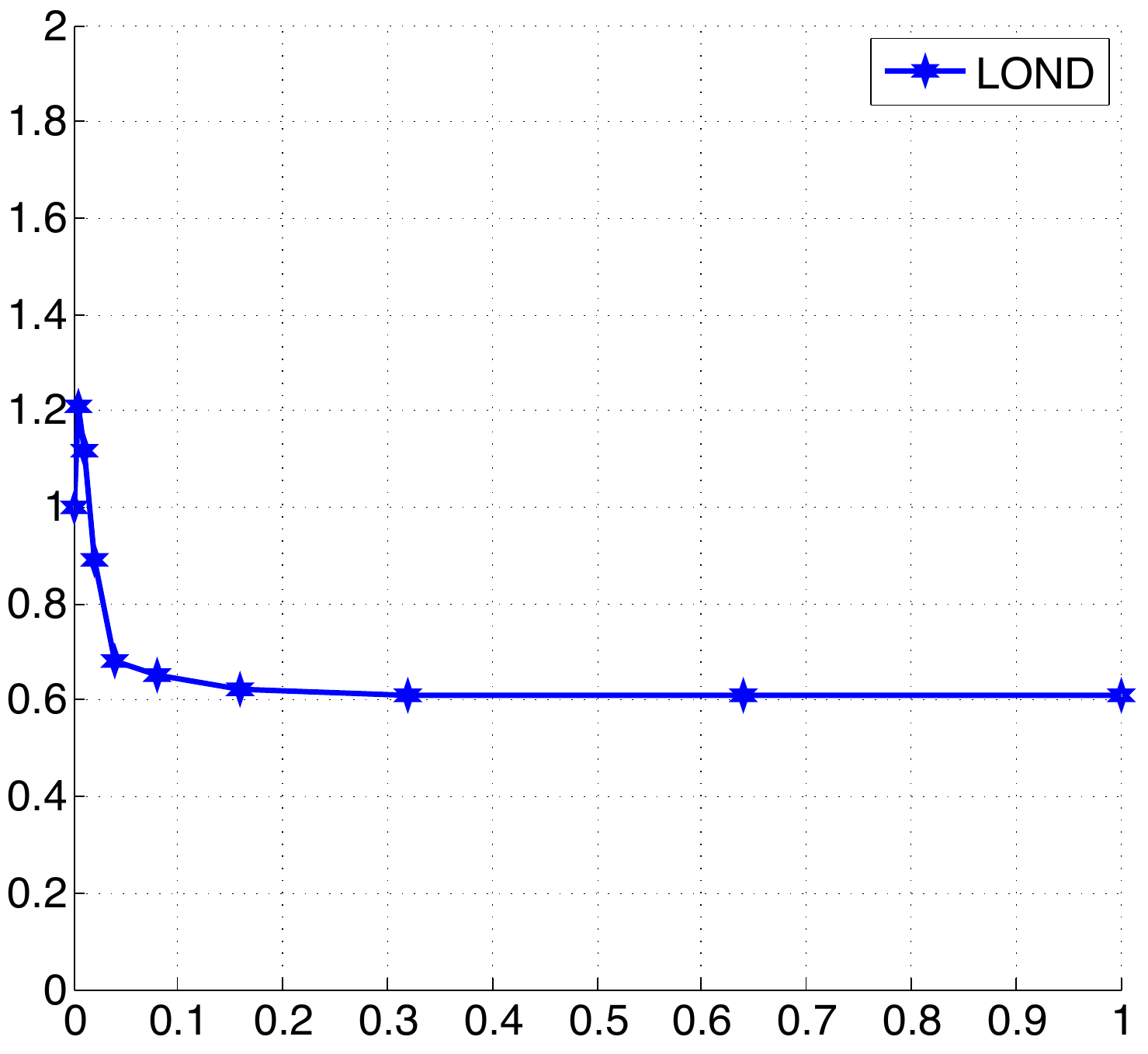}
        \put(-105,-10){$\pi$}
        \put(-105,-20){\phantom{a}}
        \put(-225,50){\rotatebox{90}{Relative power to BH}}
        \label{fig:Power_Worst_dep2}
        }
    \caption{FDR and mFDR of (adjusted) \LOND and BH under setting described in Section~\ref{sec:numerical-dep} (Scenario I). Figure (c) shows the relative power of \LOND to BH method.}\label{fig:Exp4_Worst}
    \vspace{-.7cm}
\end{figure}

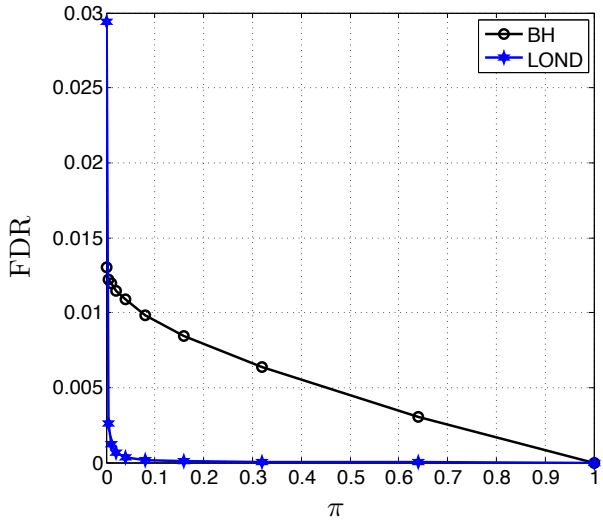
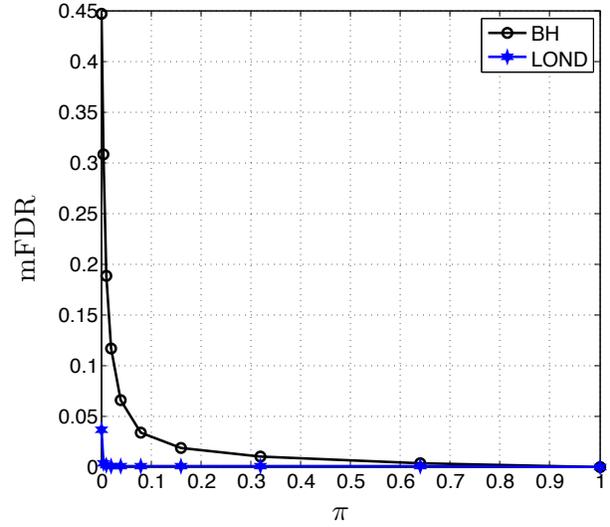
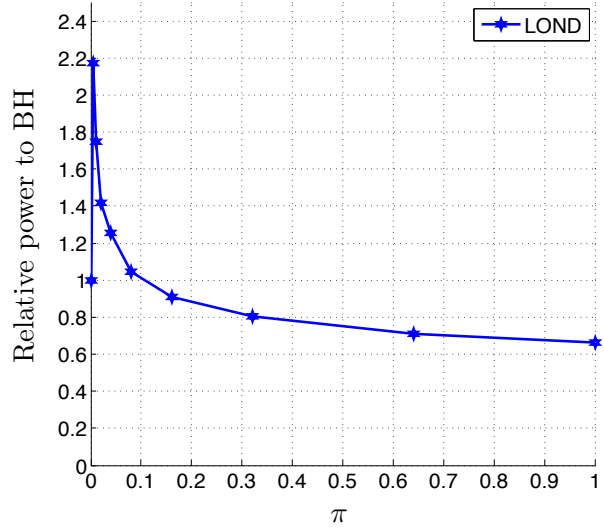
\begin{figure}[!t]
    \centering
    \subfigure[FDR]{
        \includegraphics[width = 2.9in]{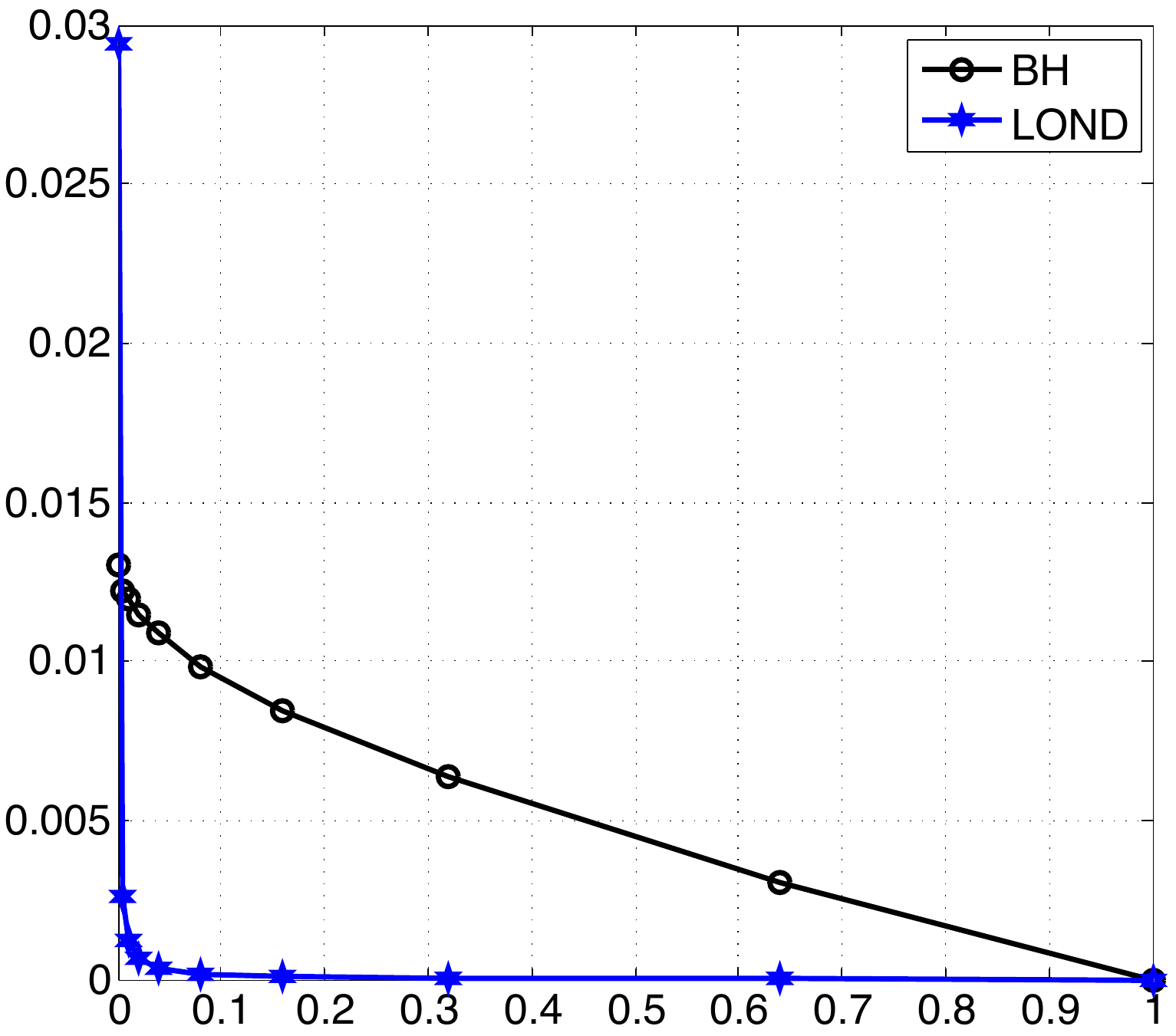}
        \put(-105,-10){$\pi$}
        \put(-105,-20){\phantom{a}}
        \put(-225,80){\rotatebox{90}{$\FDR$}}
        \label{fig:FDR_Best_dep2}
        }
        \hspace{0.8cm}
    \subfigure[mFDR]{
        \includegraphics[width=2.9in]{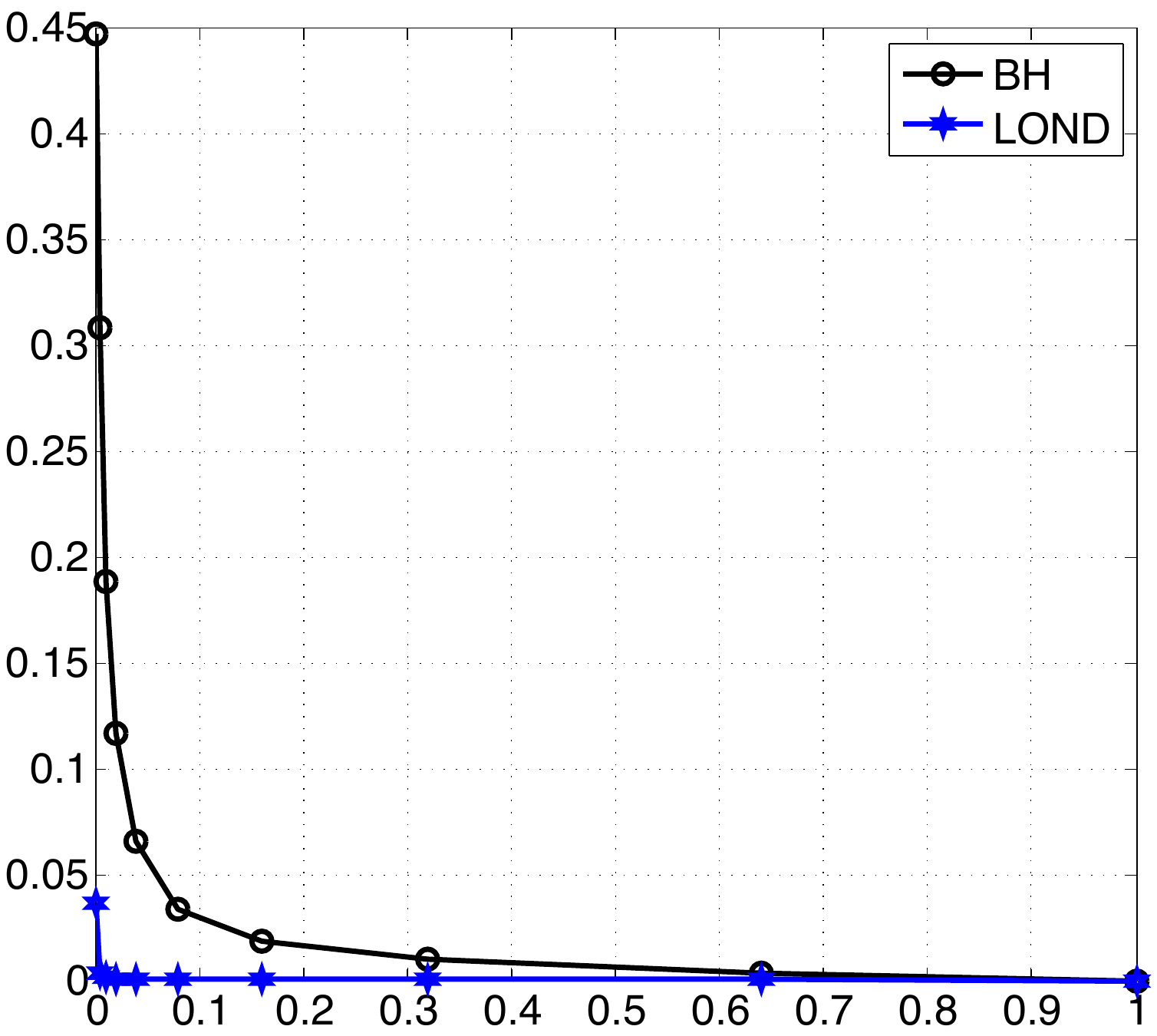}
        \put(-105,-10){$\pi$}
        \put(-105,-20){\phantom{a}}
        \put(-225,80){\rotatebox{90}{$\mFDR$}}
        \label{fig:mFDR_Best_dep2}
        }
     \subfigure[Relative power to BH]{
        \includegraphics[width=2.9in]{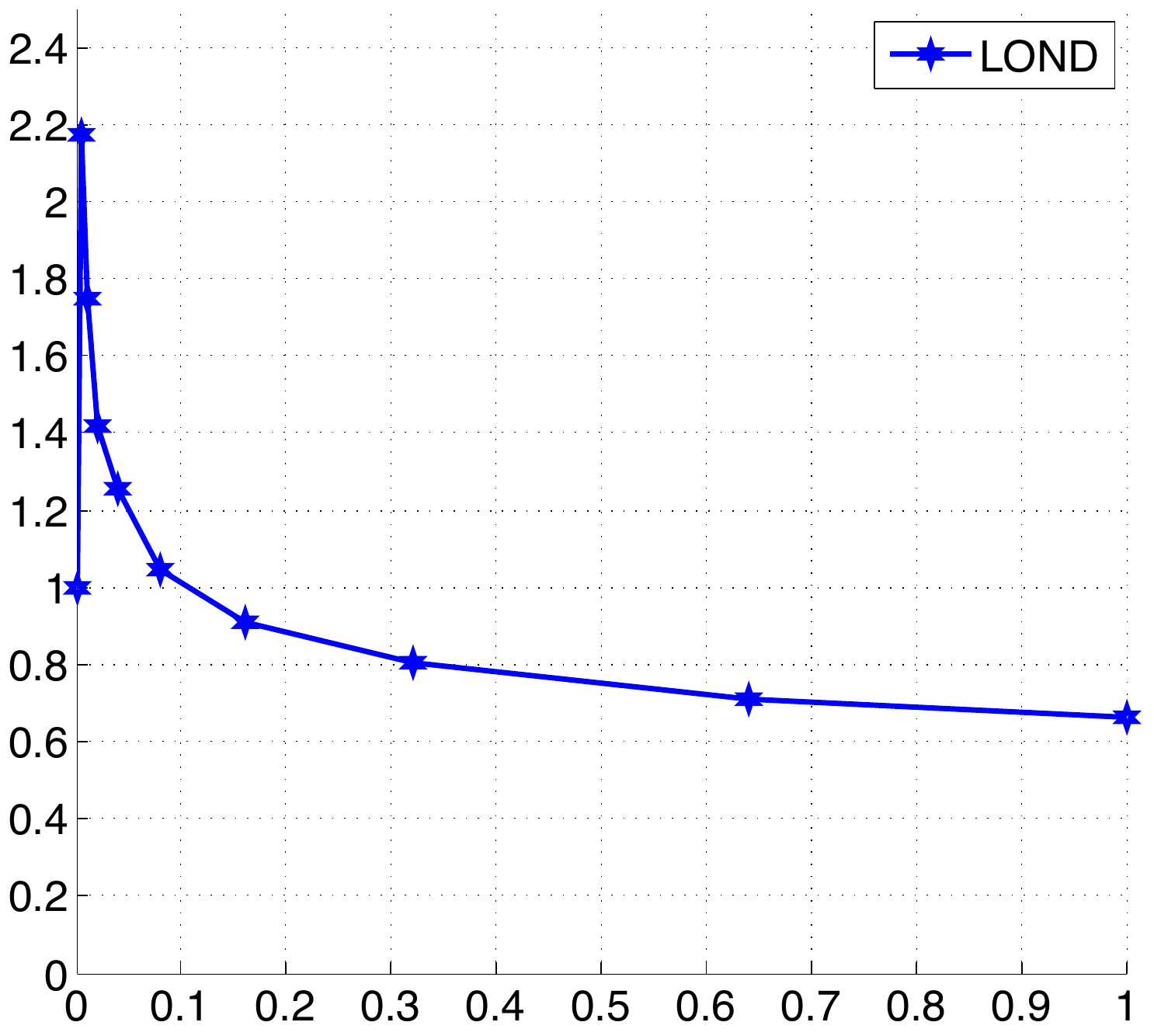}
        \put(-105,-10){$\pi$}
        \put(-105,-20){\phantom{a}}
        \put(-225,50){\rotatebox{90}{Relative power to BH}}
        \label{fig:Power_Best_dep2}
        }
    \caption{FDR and mFDR of (adjusted) \LOND and BH under setting described in Section~\ref{sec:numerical-dep} (Scenario II). Figure (c) shows the relative power of \LOND to BH method.}\label{fig:Exp4_Best}
    \vspace{-.7cm}
\end{figure}
%
\subsubsection{Dependent $p$-values}\label{sec:numerical-dep}
We evaluate the performance of $\LOND$ algorithm in case the test statistics, and thus the $p$-values, are dependent.
Here, we apply the adjustment to \LOND, as described in Section~\ref{sec:dependence}, to address dependency. 
We follows the same setting as in Section~\ref{sec:numerical-indep}, except that for each trial, the test statistics $Z = (Z_1,\dotsc, Z_n)$ are generated according to $Z\sim \normal(\theta, \Sigma)$, where $\theta = (\theta_1,\dotsc, \theta_n)$ and $\Sigma\in \reals^{n\times n}$ is constructed as follows.
Define
\begin{eqnarray*}
\tilde{\Sigma}_{ij} = \begin{cases}
1\, & \text{ if }i=j\,,\\
0.5\,&\text{otherwise.}
\end{cases}
\end{eqnarray*}
We let
\[\Sigma = \Lambda \hSigma \Lambda\,,\]
where $\Lambda$ is a diagonal matrix with uniformly random signs on the diagonal.
Figure~\ref{fig:Exp4_Worst} shows the performance of $\LOND$ and (adjusted) BH\footnote{Since the off-diagonal entries of $\Sigma$ can be negative, it does not satisfy PRDS~\cite{benjamini2001control} and in general we need the adjustment described in Section~\ref{sec:dependence} to BH to cope with dependency.} in controlling $\FDR$ and $\mFDR$ under Scenario I. 

Interestingly, in this setting the difference between criterion $\FDR$ and $\mFDR$ is more pronounced. As we see while the adjusted BH controls \FDR to be below $\alpha=0.05$, the corresponding $\mFDR$ can be as high as $0.45$. The (adjusted) \LOND though controls both \FDR and \mFDR to be less than $\alpha = 0.05$. 

Figure~\ref{fig:Exp4_Best} shows the performance of $\LOND$ and BH under Scenario II. 

To the best of our knowledge, there is no alpha-investing rule that controls $\mFDR$ in presence of dependency among the $p$-values,
and hence we do not include its evaluation in this experiment.
\subsection{Microarray example}
We illustrate performance of \LOND algorithm on a microarray example wherein the genes
to be tested arrive in the stream in an online fashion and we would like to find the significant ones while controlling $\FDR$ over infinite time horizon. 

We use a prostate cancer data set~\cite{singh2002gene}, which  
contains genetic expression levels for $n = 12600$ genes on $ m= 102$ men, $m_1 = 50$ from normal controls and $m_2 = 52$ from prostate cancer patients. Each gene yields a two-sample $t$-statistic $t_i$ comparing tumor vs non-tumor cases as follows. Let $x_{ij}$ be the expression level of gene $i$ on patient $j$. Further let $\bx_i(1)$ and $\bx_i(2)$ denote the average of $x_{ij}$ for the normal controls and for the cancer patients. The two-sample $t$-statistic for testing gene $i$ is given by
\begin{align*}
t_i = \frac{\bx_i(2) - \bx_i(1)}{s_i}\,.
\end{align*} 
Here $s_i$ is an estimate of the variance of $\bx_i(2)-\bx_i(1)$,
$$ s_i^2 = \frac{1}{m-2}\, \Big(\frac{1}{m_1}+\frac{1}{m_2}\Big) \Big\{\sum_{j\in {\text{non-tumor}}} (x_{ij} - \bx(1))^2+ \sum_{j\in \text{tumor}} (x_{ij} - \bx(2))^2\Big\}\,.$$
The two-sided $p$-value for testing the null hypothesis $H_{i}:$ ``gene $i$ is null" is then given by $P_i = F_{m-2}(t_i)$, where $F_{m-2}$ is the cumulative distribution function of a student's $t$ distribution with $m-2$ degrees of freedom.

Since the gene expressions and thus the corresponding $p$-values are dependent in general, we apply adjusted BH and \LOND, as described in Section~\ref{sec:dependence}. For significance level $\alpha = 0.05$, BH returns $459$ genes as significant and \LOND returns $203$ significant ones. The significant genes suggested by \LOND are a subset of those returned by BH. As we see, although \LOND is designed to control FDR in an online manner, it recovers a good proportion of genes discovered by BH having access to all $p$-value a priori. 
\section{Proof of theorems and lemmas}\label{sec:proofs}
\subsection{Proof of Theorem~\ref{thm:FDRII}}
We first show that $\FDR(n) \le \alpha$. Write
\begin{eqnarray}
\FDR(n) \equiv \E\Big(\frac{V^\theta(n)}{D(n)\vee 1}\Big)
 = \E\Big(\sum_{j=1}^n\frac{F^\theta_j}{D(n)\vee 1}  \Big)
 = \E\Big(\sum_{j=1}^n \frac{F^\theta_j}{\alpha_j} \cdot\frac{\alpha_j}{D(n)\vee1} \Big)\label{eq:FDRII-1}
\end{eqnarray}
%
Clearly, $(D(n)\vee 1) \ge (D(j) \vee 1)$ and by the rule~\eqref{eq:RuleII} we obtain
\begin{eqnarray}
\frac{\alpha_j}{D(n)\vee 1} \le \frac{\alpha_j}{D(j-1)+ 1} \le \beta_j\,.\label{eq:FDRII-2}
\end{eqnarray}
Condition~\eqref{eq:condition} can be stated as
\begin{eqnarray*}
\forall \theta\in \Theta,\quad \E_\theta(F^\theta_j -\alpha_j \vert D{(j-1)} )\le 0\,.
\end{eqnarray*}
Therefore,
\begin{eqnarray}
\E\Big(\frac{1}{\alpha_j} F^\theta_j\Big)= \E\Big\{\E\Big(\frac{F^\theta_j}{\alpha_j}\Big|D(j-1) \Big)\Big\} \le \E(1) =1\,,\label{eq:FDRII-3}
\end{eqnarray}
where we used the fact that $\alpha_j$ is deterministic given $D(j-1)$. 
Applying~\eqref{eq:FDRII-2} and~\eqref{eq:FDRII-3} to equation \eqref{eq:FDRII-1}, we get

$$\FDR(n) \le \sum_{j=1}^\infty \beta_j = \alpha\,.$$

We next prove that $\mFDR(n)\le \alpha$. 
Note that for all $\theta \in \Theta$, we have

$$\E_\theta(F^\theta_i - \alpha_i ) = \E(\E_\theta(F^\theta_i -\alpha_i|D(i-1))) \le 0\,,$$

where we used Condition~\eqref{eq:condition} in the last step.
Adding these inequalities for $i=1, \dotsc, n$
 we get that the following holds true for all $\theta\in \Theta$:
\begin{eqnarray}\label{eq:VnII}
\E_\theta(V^\theta(n))- \sum_{i=1}^n \E_\theta(\alpha_i)  \le 0\,.
\end{eqnarray}
Further, for all $\theta\in \Theta$ we have
\begin{align}
\sum_{i=1}^n \E_\theta(\alpha_i) &= \sum_{i=1}^n \E_\theta\{\beta_i (D(i-1)+1)\}
=\sum_{i=1}^n \E_\theta(\beta_i \sum_{j=1}^{i-1} D_j ) + \sum_{i=1}^n \beta_i \nonumber\\
& = \sum_{j=1}^{n-1} \Big(\sum_{i=j+1}^n \beta_i \Big) \E_\theta(D_j) + \sum_{i=1}^n \beta_i \nonumber\\
&\le \Big(\sum_{i=1}^n \beta_i \Big)\Big\{\E_\theta(D(n))+1\Big\}\,.\label{eq:DmII}
\end{align}
The result follows by combining equations~\eqref{eq:VnII} and~\eqref{eq:DmII}.
\subsection{Proof of Theorem~\ref{thm:FDR}}
\label{proof-thm:FDR}

We first prove the claim on controlling $\mFDR$. Proceeding along the same lines as the proof of equation~\eqref{eq:VnII}, we have
the following for all $\theta \in \Theta$:
\begin{eqnarray}\label{eq:Vn}
\E_\theta(V^\theta(n))- \sum_{i=1}^n \E_\theta(\alpha_i) \le 0\,.
\end{eqnarray}
By equation~\eqref{eq:Rule}, it is clear that for all $\theta\in \Theta$, we have
\begin{eqnarray}\label{eq:Dm}
 \alpha\E_\theta (D(n)+1) - \sum_{i=1}^n \E_\theta(\alpha_i)  \ge 0\,,
\end{eqnarray}
because sum of significance levels $\alpha_i$ between two consecutive discoveries is bounded by $\alpha$.
Combining equations~\eqref{eq:Vn} and~\eqref{eq:Dm}, we obtain the desired result.

We next prove equation~\eqref{eq:FDR}. 
Define random variables $X_i$ for $i \ge 1$ as follows:

%
\begin{equation}\label{eq:X_i}
X_i  = \begin{cases}
1& \text{ if $\tau_i < \infty$ and $i$-th discovery is false,}\\
0& \text{ if $\tau_i < \infty$ and $i$-th discovery is true,}\\
*& \text{ otherwise.}
\end{cases}
\end{equation}
%
We first show that $\E(X_i \, \ind(\tau_i < \infty) |\tau_{i-1})\le \alpha$. 
This is in fact the probability that starting from $\tau_{i-1}+1$ at least one discovery
occurs and the first discovery is false. Therefore, by applying union bound we have
\begin{eqnarray}
\E(X_i \, \ind(\tau_i < \infty) |\tau_{i-1})\le \sum_{\ell=\tau_{i-1}+1}^\infty \prob_{\theta_\ell=0}(D_\ell = 1) \le \sum_{r=1}^\infty \beta_r = \alpha\,,
\end{eqnarray}
where the last inequality follows from the fact that under null hypothesis, $\theta_\ell =0$, we have $p_\ell\sim\cU([0,1])$, and
thus $D_\ell= 1$ with probability at most $\alpha_\ell$. Moreover, $\alpha_\ell = \beta_{\ell-\tau_{i-1}}$ for $\ell \le \tau_i$ due to rule~\eqref{eq:Rule}. Hence,
\begin{align*}
\E(\FDP(\tau_k) \ind(\tau_k<\infty)) &= \E\Big[\frac{V^\theta(\tau_k)}{k}\, \ind(\tau_k< \infty) \Big]\\
 &= \frac{1}{k}\sum_{i=1}^k \E ( X_i \,\ind(\tau_k < \infty))\\
&\le \frac{1}{k}\sum_{i=1}^k \E \Big\{ \E(X_i\, \ind(\tau_i < \infty)|\tau_{i-1})\Big\} \le \alpha\,.
\end{align*}
%

\subsubsection{Proof of Remark~\ref{rem:mFDR_strong}}
Let $X_i$ be defined as per equation~\eqref{eq:X_i}. We have
\begin{eqnarray}
V^\theta(n) = \sum_{i=1}^\infty X_i\, \ind(\tau_i \le n)\,.
\end{eqnarray}
Therefore,
\begin{eqnarray}
\E (V^\theta(n)) = \sum_{i=1}^\infty \E(X_i\, \ind (\tau_i\le n))\,.
\end{eqnarray}
%
By applying union bound, we have
\begin{align*}
\E(X_i\,\ind(\tau_i\le n)|\tau_{i-1}) &\le \Big(\sum_{\ell=\tau_{i-1}+1}^n \prob_{\theta_\ell=0} (D_\ell = 1)\Big)\, \ind(\tau_{i-1}\le n-1) \\
&\le \alpha\, \ind(\tau_{i-1}\le n-1)\,.
\end{align*}
Hence, adopting the convention $\tau_0=0$, we get
\begin{align}
\E (V^\theta(n)) &\le \alpha\, \sum_{i=1}^\infty  \prob(\tau_{i-1}\le n-1) = 
\alpha\, \sum_{i=1}^\infty \prob(\tau_{i}\le n-1) + \alpha\,\nonumber\\
&= \alpha\, \E\Big(\sum_{i=1}^\infty \ind(\tau_i\le n-1) \Big) + \alpha\nonumber\\
&= \alpha\, \E (D(n-1)) + \alpha\,.\label{eq:EFn}
\end{align}
The result follows.
%
%
\subsection{Proof of Theorem~\ref{thm:dependency-FDR}}
\label{proof:dependency-FDR}
Let $\event_{v,u}$ be the event that $\LOND$ makes $v$ false and $u$ true discoveries in $\cH(n)$.
We further denote by $n_0$ and $n_1$ the number of true and false null hypotheses in $\cH(n)$.
The $\FDR(n)$ is then
$$\FDR(n) \equiv \E(\FDP(n)) = \sum_{v=0}^{n_0} \sum_{u=0}^{n_1} \frac{v}{(v+u)\vee 1}\, \prob(\event_{v,u})\,.$$
We use the following lemma from~\cite{benjamini2001control} and state its proof here for the reader's convenience.
\begin{lemma}[~\cite{benjamini2001control}]\label{lem:dummy_ind}
The following holds true:
$$\prob(\event_{v,u}) = \frac{1}{v} \sum_{i=1}^{n_0} \prob((p_i\le \alpha) \cap \event_{v,u})\,.$$
\end{lemma} 

\begin{proof}
Fix $v$ and $u$ and for a subset $\omega \in \{1, \dotsc, n_0\}$ with $|\omega| = v$, denote by $\event^\omega_{v,u}$
the event that the $v$ false discoveries are $\omega$. We further note that 
\begin{align*}
\prob((p_i \le \alpha_i) \cap \event^\omega_{v,u}) = \begin{cases}
\prob(\event^\omega_{v,u})\, & \text{ if }i\in \omega\,,\\
0& \text{ otherwise .}
\end{cases} 
\end{align*}
Therefore,
\begin{align*}
\sum_{i=1}^{n_0} \prob((p_i\le \alpha_i) \cap \event_{v,u}) &= \sum_{i=1}^{n_0} \sum_{\omega} \prob((p_i\le \alpha_i) \cap \event^\omega_{v,u})\\
&= \sum_{\omega} \sum_{i=1}^{n_0}\prob((p_i\le \alpha_i) \cap \event^\omega_{v,u})\\
&= \sum_{\omega} \sum_{i=1}^{n_0} \ind(i\in \omega) \prob(\event^\omega_{v,u}) = \sum_{\omega} v \prob(\event^\omega_{v,u})
= v \prob(\event_{v,u})\,,
\end{align*}
which completes the proof.
\end{proof}

Applying Lemma~\ref{lem:dummy_ind} we obtain
\begin{align}
\FDR(n) \equiv \sum_{v=0}^{n_0} \sum_{u=0}^{n_1}\frac{1}{(v+u)\vee1} \sum_{i=1}^{n_0} \prob((p_i\le \alpha_i) \cap \event_{v,u})\,.
\end{align}
For $i\ge 1$ and $0\le a\le i-1$, we let the event $\cC^{(i)}_{s,v,u}$ be the event that if $p_i \le \alpha_i$ (i.e., $H_i$ is rejected) then there are $u$ true discoveries, and $v$ false discoveries in $\cH(n)$ such that $s$ number of false discoveries occur before time step $i$. Clearly,
$$\prob((p_i\le \alpha_i) \cap\event_{v,u}) = \cup_{s=0}^{i-1}\, \prob((p_i\le \alpha_i) \cap \cC^{(i)}_{s,v,u})\,.$$
Since the events $\cC^{(i)}_{s,v,u}$ are disjoint for different values of $s$, we obtain
\begin{align}
\FDR(n) \equiv \sum_{v=1}^{n_0} \sum_{u=0}^{n_1}\frac{1}{v+u} \sum_{i=1}^{n_0} \sum_{s=0}^{i-1} \prob((p_i\le \alpha_i) \cap \cC^{(i)}_{s,v,u})\,.
\end{align}
Rearranging the terms and recalling the rule $\alpha_i = \tbeta_i(D(i-1)+1)$, we arrive at
\begin{align}\label{eq:isvu}
\FDR(n) \equiv \sum_{i=1}^{n_0}  \sum_{s=0}^{i-1} \sum_{v=1}^{n_0} \sum_{u=0}^{n_1}\frac{1}{v+u}  \prob([p_i\le \tbeta_i (s+1)] \cap \cC^{(i)}_{s,v,u})\,.
\end{align}
For $s+ 1\le k\le n$, define the event $\cC^{(i)}_{s,k}$ as 
$$\cC^{(i)}_{s,k} \equiv \underset{\substack{v,u \\ v+u = k}}{\bigcup} \cC^{(i)}_{s,v,u}\,.$$
In words, $\cC^{(i)}_{s,k}$ is the event that there are $k$ discoveries in $\cH(n)$ with $s$ false discoveries occurred before time $i$. 
Writing the RHS of equation~\eqref{eq:isvu} in terms of $k$ instead of $v$ and $s$, we have
\begin{align*}
\FDR(n) &\equiv \sum_{i=1}^{n_0}  \sum_{s=0}^{i-1} \sum_{k=s+1}^n \frac{1}{k} \prob([p_i\le \tbeta_i (s+1)] \cap \cC^{(i)}_{s,k})\\
&\le \sum_{i=1}^{n_0} \sum_{s=0}^{i-1}  \frac{1}{s+1}\sum_{k=s+1}^n \prob([p_i\le \tbeta_i (s+1)] \cap \cC^{(i)}_{s,k})\,.
\end{align*}
We define $\cC^{(i)}_s$  as the event that $s$ false discoveries occur before time $i$. Equivalently, 
$$\cC^{(i)}_s \equiv \underset{s+1\le k\le n}{\bigcup} \cC^{(i)}_{s,k}\,.$$ 
Given that events $\cC^{(i)}_{s,k}$ are disjoint for different values of $k$, we get
\begin{align*}\label{eq:FDR-bound}
\FDR(n) &\le \sum_{i=1}^{n_0} \sum_{s=0}^{i-1}  \frac{1}{s+1} \prob([p_i\le \tbeta_i (s+1)] \cap \cC^{(i)}_{s})\,. 
\end{align*}
For $j\in \{0,\dotsc, s\}$, denote 
$$q_{i,j,s}\equiv \prob(\{p_i\le [\tbeta_i j, \tbeta_i (j+1)]\} \cap \cC^{(i)}_{s})\,.$$
Writing bound~\eqref{eq:FDR-bound} in terms of $q_{i,j,s}$, we have
\begin{align*}
\FDR(n) &\le  \sum_{i=1}^{n_0} \sum_{s=0}^{i-1} \frac{1}{s+ 1} \sum_{j=0}^s q_{i,j,s}\\
&\le \sum_{i=1}^{n_0}\sum_{j=0}^{i-1} \sum_{s=j}^{i-1} \frac{1}{s+ 1} q_{i,j,s}
\le \sum_{i=1}^{n_0}\sum_{j=0}^{i-1}  \frac{1}{j+1}\sum_{s=j}^{i-1} q_{i,j,s}\\
& = \sum_{i=1}^{n_0}\sum_{j=0}^{i-1} \frac{1}{j+1} \prob\Big(p_i\le [\tbeta_i j, \tbeta_i (j+1)]\Big)\\
&\stackrel{(a)}{=}  \sum_{i=1}^{n_0} \Big(\sum_{j=1}^{i} \frac{1}{j}\Big) \tbeta_i \le \sum_{i=1}^{n_0} \beta_i = \alpha\,,
\end{align*}
where in step $(a)$, we used the fact $p_i\sim \cU([0,1])$ under the null hypothesis $H_i$. 
\subsection{Proof of Theorem~\ref{thm:LOND_rate}}
\label{proof:LOND_rate}
We let $\tau_\ell$ denote the time of $\ell$-th discovery. Under the mixture model, for $m\ge \tau_\ell+1$ we have
\begin{eqnarray}\label{eq:G1}
\prob(\tau_{\ell+1} > m | \tau_\ell) = \prod_{i=\tau_\ell +1}^m \Big(1- G((\ell+1) \beta_i) \Big)
\le \exp\Big\{-\sum_{i=\tau_\ell+1}^m G((\ell+1)\beta_i) \Big\}\,.
\end{eqnarray}
%
Recall that $G(x) = (1-\eps)x+ \eps F(x)$ and note that $\tau_\ell \ge \ell$. Hence, for some constant $L_0$ and all $\ell > L_0$ the following holds:
 \begin{align}
 \sum_{i=\tau_\ell+1}^m G((\ell+1)\beta_i)  &\ge \eps \lambda C^\kappa (\ell+1)^\kappa \sum_{i=\tau_\ell+1}^m  i^{- \kappa\nu} \\
 &\ge \eps \lambda C^\kappa (\ell+1)^\kappa m^{-\kappa\nu}  \sum_{i=\tau_\ell+1}^m 1 \nonumber \\
 &=  \eps \lambda C^\kappa  (\ell+1)^\kappa m^{-\kappa\nu} (m-\tau_\ell-1)  \,. \label{eq:G2}
 \end{align}
 %
 %
 %
 Combining equations~\eqref{eq:G1} and \eqref{eq:G2} we get
\begin{eqnarray*}
\prob(\tau_{\ell+1} > m|\tau_\ell) \le \exp\Big\{-\eps \lambda C^\kappa(\ell+1)^\kappa  (m - \tau_\ell-1) m^{-\kappa\nu}\Big\}\,.
\end{eqnarray*} 

 Next we compute $\E(\tau_{\ell+1}|\tau_\ell)$. Since $\prob(\tau_{\ell+1} \ge \tau_\ell+1|\tau_\ell) = 1$, we have
 \begin{align*}
 \E(\tau_{\ell+1}|\tau_\ell) &= \tau_{\ell}+1 + \sum_{m = \tau_\ell+1}^\infty \prob(\tau_{\ell+1} > m|\tau_{\ell})\\
 &\le \tau_{\ell}+2 + \sum_{i=1}^\infty \exp\Big\{-\eps \lambda C^\kappa (\ell+1)^\kappa i (i+\tau_\ell+1)^{-\kappa \nu} \Big\}
 \end{align*}
 We split the summation over $1\le i\le \tau_{\ell}+1$ and $\tau_{\ell}+2 \le i$ and upper bound each term separately.
 \begin{align}
 I_1 &\equiv  \sum_{i=1}^{\tau_{\ell}+1} \exp\Big\{-\eps \lambda C^\kappa (\ell+1)^\kappa i (i+\tau_\ell+1)^{-\kappa \nu} \Big\}\nonumber\\
 &\le \sum_{i=1}^{\tau_{\ell}+1} \exp\Big\{-\eps \lambda C^\kappa (\ell+1)^\kappa i (2\tau_\ell+2)^{-\kappa \nu} \Big\}\nonumber\\
 &\stackrel{(a)}{\le} \int_0^{\tau_\ell+1}  \exp\Big\{- \eps \lambda C^\kappa (\ell+1)^\kappa (2\tau_\ell+2)^{-\kappa \nu} z \Big\} \de z\nonumber\\
 & \le \frac{(2\tau_\ell+2)^{\kappa \nu}}{\eps \lambda C^\kappa (\ell+1)^\kappa}\,,\label{eq:I1}
 \end{align}
 where $(a)$ holds since the summand is decreasing in $i$. we next bound the second term as follows.
  \begin{align}
 I_2 &\equiv  \sum_{i=\tau_\ell+2}^{\infty} \exp\Big\{-\eps \lambda C^\kappa (\ell+1)^\kappa i (i+\tau_\ell+1)^{-\kappa \nu} \Big\}\nonumber\\
 &\le \sum_{i=\tau_\ell+2}^\infty \exp\Big\{-\eps \lambda C^\kappa (\ell+1)^\kappa  2^{-\kappa \nu} i^{1-\kappa\nu} \Big\}\nonumber\\
 &\stackrel{(a)}{\le} \int_{\tau_\ell+1}^\infty  \exp\Big\{-\eps \lambda C^\kappa (\ell+1)^\kappa 2^{-\kappa \nu} z^{1-\kappa\nu} \Big\} \de z\,,\label{eq:I2}
 \end{align}
 where $(a)$ holds since the summand is decreasing in $i$. Define $C_* = \eps \lambda C^\kappa (\ell+1)^\kappa 2^{-\kappa\nu}$. It is straightforward to see that
 \begin{eqnarray}
 \int_{0}^\infty  \exp\Big(-C_* z^{1-\kappa\nu} \Big) \de z =
 \frac{1}{1-\kappa\nu} C_*^{-1/(1-\kappa\nu)} \Gamma\Big(\frac{1}{1-\kappa\nu}\Big)\,.\label{eq:I2-2}
 \end{eqnarray}
 Combining the bounds~\eqref{eq:I1}, \eqref{eq:I2} and~\eqref{eq:I2-2}, we obtain that for $\ell >L_0$,
 \begin{eqnarray}
 \E(\tau_{\ell+1} | \tau_\ell) \le \tau_\ell+2+ \frac{(2\tau_\ell+2)^{\kappa \nu}}{\eps \lambda C^\kappa (\ell+1)^\kappa}+ C' (\ell+1)^{-\frac{\kappa}{1-\kappa\nu}}\,,\label{eq:recursI}
 \end{eqnarray}
 with 
 $$C'\equiv \frac{1}{1-\kappa\nu}\,\Big(\eps \lambda C^\kappa 2^{-\kappa\nu}\Big)^{{-1}/(1-\kappa\nu)} \Gamma\Big(\frac{1}{1-\kappa\nu}\Big)\,.$$
 Here $\Gamma(t) = \int_0^\infty z^{t-1}e^{-z} \de z$ is the gamma function.
 
 Let $\lambda_\ell \equiv \E(\tau_\ell)$. Since $f(s) = s^{\kappa \nu}$ is concave, by applying Jensen inequality to equation~\eqref{eq:recursI}, we arrive at the following recursive bound for $\ell > L_0$:
 \begin{eqnarray}\label{eq:lambda_rec}
 \lambda_{\ell+1} \le \lambda_\ell+2 + \frac{(2\lambda_\ell+2)^{\kappa \nu}}{\eps \lambda C^\kappa (\ell+1)^\kappa}+ C' (\ell+1)^{\frac{-\kappa}{1-\kappa\nu}}\,.
 \end{eqnarray}
 In Appendix~\ref{app:lambda_B}, we show that using equation~\eqref{eq:lambda_rec}
 \begin{eqnarray}\label{eq:lambda_B}
 \lambda_\ell \le \tC\,\ell^{\frac{1-\kappa}{1-\kappa\nu}}\,,
 \end{eqnarray}
 for some constant $\tC = \tC(\kappa,\nu)$. 
 Let $\zeta_n \equiv (\delta n/\tC)^{\frac{1-\kappa\nu}{1-\kappa}}$.
 Applying Markov inequality, for any $n$ and any fixed $0<\delta<1$, 
 we obtain
 \begin{align}
 \prob\{D^\LOND(n) < \zeta_n\}  &= \prob\{\tau_{\zeta_n} > n\} \le \frac{\lambda_{\zeta_n}}{n}
 \le \frac{\tC\zeta_n^{\frac{1-\kappa}{1-\kappa\nu}}}{n}  = \delta\,,
 \end{align}
which completes the proof.

\appendix

\section{Derivation of equation~\eqref{eq:lambda_B}}\label{app:lambda_B}
Let $C_1 = 2^{\kappa \nu}/(\eps \lambda C^\kappa)$ and $C_2 = C'+2$. We relax the bound~\eqref{eq:lambda_rec} as
\begin{eqnarray}\label{eq:lambda_B1}
\lambda_{\ell+1} \le \lambda_\ell + C_1 \frac{(\lambda_\ell+1)^{\kappa \nu}}{(\ell+1)^\kappa} + C_2\,.
\end{eqnarray}
We write $\lambda_\ell = g_\ell+\Delta_\ell$ where 
$$g_\ell \equiv \left(C_1\frac{1-\kappa\nu}{1-\kappa}\right)^{\frac{1}{1-\kappa\nu}} (\ell+1)^{\frac{1-\kappa}{1-\kappa\nu}} -1\,.$$
Writing equation~\eqref{eq:lambda_B1} in terms of $\Delta_\ell$ and $g_\ell$ we get, for $\ell> L_0$, 
\begin{align*}
g_{\ell+1} + \Delta_{\ell+1} &\le g_\ell+\Delta_\ell + C_1 \frac{(g_\ell+\Delta_\ell+1)^{\kappa\nu}}{(\ell+1)^\kappa}  + C_2\\
&= g_\ell+\Delta_\ell + C_1 \frac{(g_\ell+1)^{\kappa\nu}}{(\ell+1)^\kappa} \Big(1+\frac{\Delta_\ell}{g_\ell+1}\Big)^{\kappa\nu} + C_2\,.
\end{align*}
Since $0 < \kappa\nu \le 1$, applying Bernoulli's inequality, we obtain
\begin{align}\label{eq:lambda_B3}
g_{\ell+1} + \Delta_{\ell+1} \le g_\ell+\Delta_\ell + C_1 \frac{(g_\ell+1)^{\kappa\nu}}{(\ell+1)^\kappa} \Big(1+ \kappa\nu\frac{\Delta_\ell}{g_\ell+1}\Big) + C_2\,.
\end{align}
Note that 
\begin{align*}
C_1\frac{(g_\ell+1)^{\kappa\nu}}{(\ell+1)^\kappa} &= C_1^{\frac{1}{1-\kappa\nu}}  \Big(\frac{1-\kappa\nu}{1-\kappa}\Big)^{\frac{\kappa\nu}{1-\kappa\nu}} (\ell+1)^{\frac{\kappa\nu-\kappa}{1-\kappa\nu}}\\
&\le  \Big(C_1\frac{1-\kappa\nu}{1-\kappa}\Big)^{\frac{1}{1-\kappa\nu}} \Big\{(\ell+ 2)^{\frac{1-\kappa}{1-\kappa\nu}} - (\ell+1)^{\frac{1-\kappa}{1-\kappa\nu}}\Big\}\\
&= g_{\ell+1} - g_\ell\,,
\end{align*}
where the inequality holds since $(1-\kappa)/(1-\kappa\nu)>1$. Using this bound in equation~\eqref{eq:lambda_B3} gives the following
\begin{align*}
\Delta_{\ell+1} \le \Delta_\ell + C_1 \kappa \nu\frac{(g_\ell+1)^{\kappa\nu-1}}{(\ell+1)^\kappa}\Delta_\ell + C_2\,.
\end{align*}
Plugging in for $g_\ell$, we arrive at
\begin{eqnarray}
\Delta_{\ell+1} \le \Delta_\ell + \frac{\kappa\nu(1-\kappa)}{1-\kappa\nu} \frac{\Delta_\ell}{\ell+1} + C_2\,,
\end{eqnarray}
for $\ell\ge L_0$. The claim follows readily applying the lemma below with $\beta = \kappa \nu$ and $\mu = (1-\kappa)/(1-\kappa\nu)$.
\begin{lemma}\label{lem:Delta}
Suppose that the sequence $\{\Delta_\ell\}_{\ell=1}^\infty$ satisfies the following inequalities for $\ell\ge L_0$ with constants $0<\beta<1$, $\mu\ge1$ and $C>0$.
\begin{align}\label{eq:lem_Delta}
\Delta_{\ell+1} \le \Delta_\ell + \beta \mu \frac{\Delta_\ell}{\ell+1} + C\,.
\end{align}
Then there exits constant $C_0>0$ such that $\Delta_\ell \le C_0 \ell^\mu$ for $\ell \ge 1$.
\end{lemma}
\begin{proof}
Choose $C_0$ as follows
\begin{align}\label{eq:C0}
C_0 \equiv \max\Big\{(\Delta_i i^{-\mu})_{1\le i\le L_0}, \frac{2C}{\mu(1-\beta)}\Big\}\,.
\end{align}
We prove the claim by induction on $\ell$. The induction basis $\ell = L_0$ holds clearly by choice of $C_0$.
Suppose that the claim holds for $\ell$; we prove it for $\ell+1$. Using equation~\eqref{eq:lem_Delta},we have
\begin{align*}
\Delta_{\ell+1} &\le \Delta_\ell + \beta \mu \frac{\Delta_\ell}{\ell+1} +C\\
 &\le C_0\ell^\mu + C_0\beta \mu\frac{\ell^\mu}{\ell+1}+ C\,. 
\end{align*}
In order to prove the claim, it is sufficient to show that the RHS above is not larger than $C_0 (\ell+1)^\mu$. Equivalently,
$$(\ell+1)\ell^\mu + \beta\mu \ell^\mu + (C/C_0) (\ell+1) - (\ell+1)^{\mu+1} \le 0\,. $$
Using inequality $(\mu+1)\ell^\mu \le (\ell+1)^{\mu+1} - \ell^{\mu+1}$, we bound the RHS as follows
\begin{align*}
&(\ell+1)\ell^\mu + \beta\mu \ell^\mu + (C/C_0) (\ell+1) - (\ell+1)^{\mu+1}\\
&\le(1+\beta\mu) \ell^\mu + (C/C_0) (\ell+1) -(\mu+1)\ell^\mu\\
&= (\beta-1)\mu\ell^\mu + (C/C_0)(\ell+1)\\
&\le (\beta-1)\mu\ell^\mu + (2C/C_0)\ell \le 0\,,
\end{align*}
where the last inequality follows from the choice of $C_0$ as per~\eqref{eq:C0}.
\end{proof}
\section{Proof of Theorem~\ref{thm:EDm}}
 Given the update rule~\eqref{eq:Rule}, it is clear that the times between successive discoveries are i.i.d., under the mixture model. Therefore, the process of discoveries is a renewal process. Let $\mu \equiv \E(t_i)$ be the mean inter-discovery time, where $t_i = \tau_i - \tau_{i-1}$. By the strong law of large numbers for renewal processes~\cite{durrett2010probability}, the following holds almost surely
 $$\lim_{n\to \infty}\frac{1}{n}\,D^\LORD(n) = \frac{1}{\mu}\,.$$
 We also have
 $$\prob(t_i \ge k)  = \prod_{\ell=1}^k \Big(1 - G(\beta_\ell) \Big) \le e^{-\sum_{\ell=1}^k G(\beta_k)}\,.$$
 Therefore,
 $$\mu \equiv \E(t_i) = \sum_{k=1}^\infty  \prob(t_i \ge k) \le \sum_{k=1}^\infty e^{-\sum_{\ell=1}^k G(\beta_k)}\,,$$
 which yields the desired result. Equation~\eqref{eq:LORD_E} follows by using the elementary renewal theorem~\cite{durrett2010probability}.
\bibliographystyle{myalpha}
\bibliography{all-bibliography}

\end{document}